\documentclass{scrartcl}[12pt]

\usepackage{amsmath}
\usepackage{amssymb}
\usepackage{amsthm}
\usepackage{graphicx}
\usepackage[ngerman, english]{babel}
\usepackage{array}
\usepackage{verbatim} 
\usepackage{paralist}
\usepackage{booktabs}
\usepackage[round]{natbib}
\usepackage{fancybox}
\usepackage[all]{xy}
\usepackage{babelbib}
\usepackage{multirow}
\usepackage{url}

\newcommand{\N}{\mathbb{N}}

\newcommand{\R}{\mathbb{R}}

\newcommand{\F}{\mathcal{F}}

\newcommand{\Omikron}{\mathcal{O}}

\newcommand{\eps}{\varepsilon}

\newcommand{\beq}{\begin{equation}}
\newcommand{\eeq}{\end{equation}}
\newcommand{\beqs}{\begin{equation*}}
\newcommand{\eeqs}{\end{equation*}}
\newcommand{\bal}{\begin{align}}
\newcommand{\eal}{\end{align}}
\newcommand{\bals}{\begin{align*}}
\newcommand{\eals}{\end{align*}}

\newcommand{\ld}{,\ldots,}
\newcommand{\bfx}{\mathbf{x}}
\newcommand{\bfb}{\mathbf{b}}

\newcommand{\bfM}{\mathbf{M}}
\newcommand{\bfN}{\mathbf{N}}
\newcommand{\bfa}{\mathbf{a}}
\newcommand{\bfe}{\mathbf{e}}
\newcommand{\bfs}{\mathbf{s}}

\newcommand{\bfR}{\mathbf{R}}
\newcommand{\bfA}{\mathbf{A}}
\newcommand{\bfr}{\mathbf{r}}
\newcommand{\bfd}{\mathbf{d}}
\newcommand{\bfv}{\mathbf{v}}
\newcommand{\bfu}{\mathbf{u}}
\newcommand{\bfw}{\mathbf{w}}

\newcommand{\bfI}{\mathbf{I}}
\newcommand{\bfL}{\mathbf{\Lambda}}
\newcommand{\bfG}{\mathbf{\Gamma}}
\newcommand{\Imax}{I^{\max}}
\newcommand{\Rgreat}{\mathbf{R}_{\rm great}}

\newcommand{\Rsmall}{\mathbf{R}_{\rm small}}
\newcommand{\rsmall}{\mathbf{r}_{\rm small}}
\newcommand{\ssmall}{\mathbf{s}_{\rm small}}
\newcommand{\sgreat}{\mathbf{s}_{\rm great}}
\newcommand{\rgreat}{\mathbf{r}_{\rm great}}

\newcommand{\bzero}{\mathbf{0}}

\newcommand{\bfMd}{\mathbf{M}^{\bfd}}
\newcommand{\bfMs}{\mathbf{M}^{\bfs}}
\newcommand{\diag}{{\rm diag}}
\newcommand{\nus}{\nu^{\bfs}}
\newcommand{\nud}{\nu^{\bfd}}

\def\bi{\begin{itemize}}
\def\ei{\end{itemize}}

\newtheorem{theorem}{THEOREM}

\newtheorem{proposition}{PROPOSITION}
\newtheorem{lemma}{LEMMA}
\newtheorem{definition}{DEFINITION}
\newtheorem{assumption}{ASSUMPTION}
\newtheorem{example}{EXAMPLE}
\newtheorem{inneralgorithm}{ALGORITHM}
\newenvironment{algorithm}[1]
  {\renewcommand\theinneralgorithm{#1}\inneralgorithm}
  {\endinneralgorithm}

\newcounter{remark}

\begin{document}

\title{Valuation Algorithms for Structural Models of Financial Interconnectedness}
\author{Johannes Hain and Tom Fischer 
\thanks{Institute of Mathematics, University of Wuerzburg, Am Hubland, 97074 Wuerzburg, Germany.
Tel: +49 931 31 84869. E-mail address of corresponding author: {\tt johannes.hain@uni-wuerzburg.de}}
\thanks{The authors thank Roger Nussbaum for background information regarding a fixed point problem.
}\\
{University of Wuerzburg}} 
\date{\today}
\maketitle

\begin{abstract}
Much research in systemic risk is focused on default contagion. While this demands
an understanding of valuation, fewer articles specifically deal with the existence, the uniqueness, and 
the computation of equilibrium prices in structural models of interconnected
financial systems. However, beyond contagion research, these topics are also essential for risk-neutral pricing. 
In this article, we therefore study and compare valuation algorithms in the standard model of
debt and equity cross-ownership which has crystallized in the work of several authors over the past one
and a half decades. Since known algorithms have potentially infinite runtime,
we develop a class of new algorithms, which find exact solutions in finitely
many calculation steps. A simulation study for a range of financial system designs allows us to derive
conclusions about the efficiency of different numerical methods under different system parameters.
\end{abstract}




\pagebreak
\tableofcontents

\section{Introduction}

Since the turn of the millennium, research interest in systemic financial risk has
steadily grown, with a noticeable pick-up in the number of publications over the past five
years. One main field of interest is default contagion, see for instance the works
of \citet{acemoglu13}, \citet{elliott13}, \citet{gai11} and \citet{nier07},
or \citet{staum12} for a survey. Many publications regarding systemic risk refer to the seminal
work of \citet{eisenberg01}, who were the first to structurally model financial systems in which
firms can hold each other's financial obligations as assets under the assumption of limited liability. 
However, the core idea behind
the model of \citeauthor{eisenberg01}, which can be interpreted as a multi-firm extension of the
\citet{merton74} model where cross-holdings of zero-coupon debt between members of the
system is allowed, received somewhat less attention by the research community. The main difference
between Eisenberg and Noe and the standard multi-firm Merton model is that prices at maturity are
not trivially determined since the value of one firm's equity
or debt may depend on the value of the debt of any other firm in the system.
\citet{eisenberg01} gave conditions under which only one equilibrium solution
exists at maturity. Together with a finite numerical algorithm that they provided, the
model could not only be used for default contagion research, but also for risk-neutral
no-arbitrage valuation under financial interconnectedness \citep[see also][]{fischer14}.
\citet{elsinger09} generalized the Eisenberg and Noe setup by also including cross-holdings in
equity, and by allowing a seniority structure of the liabilities. A numerical algorithm was provided,
however, a finite number of steps to the equilibrium price vector could not be guaranteed anymore. 

Already in \citeyear{suzuki02}, Teruyoshi Suzuki had -- unbeknown to him -- generalized the \citeauthor{eisenberg01} 
setup to the situation where debt of one single seniority and equity could be cross-owned within a
financial system (Suzuki, 2002). Unlike \citet{eisenberg01} and \citet{elsinger09}, 
\citeauthor{suzuki02}, who had the clear intention of generalizing \citet{merton74}, provided a Picard Iteration
as the numerical means of calculating price equilibria. In a further generalization of Suzuki's and
Elsinger's work, \citet{fischer14} extended the structural model of interconnectedness
to the case where liabilities could be derivatives in the sense of a dependence on other
liabilities or equities in the system. As in \citet{suzuki02}, the numerical procedure provided
to solve the liquidation value equations at maturity was the Picard Iteration. However, unfortunately, a
Picard Iteration cannot warrant an exact solution in finitely many steps.

So, while there exists a small but growing amount of research on the existence and the uniqueness of price 
equilibria in systems with financial interconnectedness, the provided algorithms mainly reflect the individual
authors' particular approach to the problem. For instance, there also exists a publication by 
\citet{gourieroux13} which mentions that, in the Elsinger model with two debt seniorities but
no equity cross-holdings, a simplex method can be applied. However, comparative studies of the different
methods seems to be absent from the existing literature. Furthermore, at present, no numerical
algorithm for the
setup with cross-holdings of equity and one seniority class of debt \citep{suzuki02,elsinger09, gourieroux12} is
known that reaches the exact solution in a finite number of calculation steps. 

For these reasons, the article at hand has three main objectives. First,
we want to provide an overview of the already existing valuation algorithms by unifying
notation and by embedding them in one general model framework. Second, we provide a
new type of algorithm which is a hybrid of Eisenberg and Noe's and Elsinger's approach that
has improved convergence properties.
Third, we introduce a whole range of algorithm versions which are based on the three
different types -- namely Picard, Elsinger, and Hybrid -- which will reach the exact solution
(if existing and unique) in a finite number of calculation steps.
Introducing these new algorithms, we show that for the three types of algorithms there always
exists an increasing and a decreasing version -- depending on a properly chosen (and explicitly
given) starting point. Furthermore, we use default set techniques and linearization techniques
to achieve finiteness. A simulation study finally allows to draw some conclusions about the
efficiency of the presented numerical algorithms with respect to model parameters such as system
size.

The structure of this paper is as follows. In the next section, we will establish the model and
necessary assumptions for a unique solution of the financial system. The existing valuation methods
of \citet{suzuki02} and \citet{elsinger09} are presented in Section \ref{sec:iterative_algo}, 
where a hybrid version of the algorithms of \citet{eisenberg01} and \citet{elsinger09} is developed as well. 
The algorithms of this section are all non-finite. 
In the fourth section, we introduce a new class of valuation algorithms based on default set techniques that find solutions in finite time.
A simulation study in Section \ref{sec:simulation} compares the runtimes of the different 
algorithms for different classes of financial systems. 
In Section \ref{sec:summary}, we conclude. A technical appendix follows.

\section{Notation and Model Assumptions}\label{sec:notation}

For two matrices $\bfM = (M_{ij})_{i,j=1\ld n}\in\R^{n\times n}$ and $\bfN = (N_{ij})_{i,j=1\ld n}\in\R^{n\times n}$ we write 
$\bfM\ge\bfN$ if $M_{ij}\ge N_{ij}$ for all $i,j\in\{1\ld n\}$ and $\bfM>\bfN$ if $M_{ij}>N_{ij}$ for at least one pair $(i,j)$. 
For two vectors $\bfu = (u_1\ld u_n)^t\in\R^n$ and $\bfv = (v_1\ld v_n)^t\in\R^n$ the definition of $\bfu\ge\bfv$ and $\bfu > \bfv$ is analogous 
to the conventions for matrices above. 
A matrix $\bfM\in\R^{n\times n}$ is said to be \emph{left substochastic} if $M_{ij}\ge 0$ for all $i,j\in\{1\ld n\}$ 
and if $\sum_{i=1}^nM_{ij}\le 1$ for all $j\in\{1\ld n\}$. 
The symbol $\bfI_n$ is used for the $n\times n$-identity matrix and $\bzero_n$ is used for a (column) vector of length $n$ 
that contains only zeros. 
Additionally, $\bzero_{n\times n}$ stands for an $(n\times n)$-matrix with only zero entries. 
For a vector $\bfu\in\R^n$, the expression $\diag(\bfu\le\bzero_n)$ stands for an $(n\times n)$-diagonal matrix 
where the $i$-th entry is 1 if $u_i\le0$ and 0 otherwise, i.e.
\begin{equation}
\diag(\bfu\le\bzero_n)  = \begin{cases} 1 ,& \text{for $i=j$ and $u_i\le0$}, \\
                                        0 ,& \text{else}. \end{cases}
\end{equation}
All operations such as the minimum, $\min\{\cdot\}$, the maximum, $\max\{\cdot\}$, or the positive part $(\cdot)^+$ 
are applied element-wise to vectors and matrices. 
The norm in this paper is the $\ell^1$-norm on $\R^n$ defined as
\begin{equation}
\|\bfx\| := \|\bfx\|_1 = \sum_{i=1}^n |x_i| \quad \text{for $\bfx\in\R^n$}.
\end{equation}
The corresponding norm for a left substochastic matrix $\bfM\in\R^{n\times n}$ is given by
\begin{equation}
\|\bfM\| := \|\bfM\|_1 = \max_{\|\bfx\|=1}\|\bfM\bfx\|_1 = \max_j \sum_{i=1}^nM_{ij} 
\le 1,
\end{equation}
meaning that $\|\bfM\|$ is the maximum of the column sums. One easily can show that $\|\bfM\bfx\|\le\|\bfM\|\|\bfx\|$.

We consider a system of $n$ financial entities, and denote $\mathcal N=\{1\ld n\}$. 
In the following these entities are simply called ``firms''. 
Each firm owns exogenous assets, that are defined in the next step.

\begin{definition}\label{def:ex_assets}
Let $a_i\ge0$ denote the market value of the \emph{exogenous assets} held by firm $i$. 
As the name implies, these assets are priced outside the considered system in the sense that 
the capital structure of the $n$ firms has no influence on the pricing mechanism of such an asset. 
By $\bfa=(a_1\ld a_n)^t\in(\R_0^+)^n$ we denote the (column) vector of the exogenous assets. 
\end{definition}

Moreover, we assume that the firms have outstanding liabilities with a nominal value at maturity of $d_i$ for each firm $i$. 
These liabilities are summarized in the vector $\bfd\in(\R_0^+)^n$. 
In our framework we assume that the entries of $\bfd$ are constant. 
Since it is assumed that the exogenous assets' prices are given by the constant vector $\bfa$, 
the results of this paper presented in the Sections \ref{sec:notation} to \ref{sec:def-set-algo} also hold 
if $\bfd$ depends on $\bfa$, i.e. if $\bfd=\bfd(\bfa)$.
However, for the remainder, we will write $\bfd$ for convenience. 
This definition of the liability vector allows the interpretation that the $d_i$  are simple loans or zero coupon bonds 
since they are not derivatives that can depend on the other assets within the system.
The case of constant liabilities is used in most existing publications in this field, 
see for example \citet{suzuki02}, \citet{gourieroux12} and \citet{elsinger09}. 
The more general case in which $\bfd$ depends for example on the endogenous assets, 
is also treated in the literature \citep[see][]{fischer14}.

To take the interconnectedness of the firms into account, we allow that each firm can own a fraction of the liabilities of the other firms.
To formalize these possible cross-holdings, we use ownership matrices. 

\begin{definition}\label{def:xos_matrix}
The left substochastic matrix $\bfMd\in\R^{n\times n}$ in which the entry $0\le M^{\bfd}_{ij}\le 1$ denotes the fraction 
that firm $i$ owns of the liability of firm $j$ is called \emph{debt ownership matrix}. 
Since no firm is allowed to hold liabilities against themselves, we assume $M^{\bfd}_{ii}=0$ for all $i\in\mathcal N$. 
The entries $M^{\bfs}_{ij}$ of the (left substochastic) \emph{equity ownership matrix} $\bfMs\in\R^{n\times n}$ 
are defined as the fraction that firm $i$ owns of firm $j$'s equity. 
\end{definition} 

Note that the diagonal entries of $\bfMs$ must not be zero; that means
it is allowed that firm $i$ holds its own shares in which case $M^{\bfs}_{ii}>0$. 
For the debt ownership matrix it is a common convention (cf. \citet{eisenberg01} or \citet{elsinger09}) 
that $M^{\bfd}_{ii}=0$ since a firm cannot have debt obligations to itself. 
The tuple $\F=(\bfa,\bfd,\bfMd,\bfMs)$ is in the following sometimes referred to as the \emph{financial system}. 

Associated with the liability vector, we consider the recovery claim vector $\bfr\in(\R_0^+)^n$. 
The recovery claim vector represents the actual payments of the firms at maturity, 
i.e. in general we have $\bfr^k\le\bfd^k$ since default risk is present. 
The value of the debt claim that firm $i$ has against firm $j$ is hence given by $M^{\bfd}_{ij} r_j$ 
and the total value of firm $i$'s debt claim against the other members of the system is $\sum_{j=1}^nM_{ij}^{\bfd} r_j$.  
Furthermore, denote by $\bfs\in(\R_0^+)^n$ the equity values of the $n$ firms which means 
that the total recovery value of all system-endogenous assets that firm $i$ owns is given by the $i$-th entry of
\begin{equation}
\bfMd\bfr + \bfMs\bfs.
\end{equation}

The basic assumption for the model is that equity is considered to be the residual claim which means that any 
outstanding liability has to be paid off completely before the shareholders receive a positive payment. 
Hence, the equity value $s_i$ of firm $i$ is positive if and only if firm $i$ can fully satisfy all their obligees. 
The Absolute Priority Rule immediately leads to the following \emph{liquidation value equations} 
for the recovery claims and the equities \citep[cf.][]{fischer14}:
\begin{align}
\label{eq:liq_eq_debt}\bfr &= \min\{\bfd, \bfa + \bfMd\bfr + \bfMs\bfs\} \\
\label{eq:liq_eq_equity}\bfs &= (\bfa + \bfMd\bfr + \bfMs\bfs - \bfd)^+,
\end{align}
where the sum in \eqref{eq:liq_eq_debt} contains no $(\cdot)^+$ since Theorem \ref{theo:unique_fp} will show 
that all solutions of this system are non-negative.
A solution for the liquidation value equations \eqref{eq:liq_eq_debt} and \eqref{eq:liq_eq_equity} is hence a fixed point of the mapping 
$\Phi:(\R_0^+)^{2n}\to (\R_0^+)^{2n}$, where $\bfR=(\bfr^t,\bfs^t)^t\in(\R_0^+)^{2n}$ and
\begin{equation}\label{eq:Phi}
\Phi(\bfR) = \Phi\begin{pmatrix}\bfr \\ \bfs \end{pmatrix} = 
\begin{pmatrix} \min\{\bfd, \bfa + \bfMd\bfr + \bfMs\bfs\} \\ (\bfa + \bfMd\bfr + \bfMs\bfs - \bfd)^+ \end{pmatrix}.
\end{equation}
We will sometimes refer to the debt component of $\bfR$ and mean in such cases the first $n$ components of $\bfR$ 
that represent the debt payments of the systems. 
The components $n+1$ to $2n$ of $\bfR$ we also call equity component for the same reasons. 
We are interested in finding the fixed points of $\Phi$, which we will also call solutions of the financial system $\F=(\bfa,\bfd,\bfMd,\bfMs)$. 
Without further constraints it is possible that there exist several fixed points. 
To ensure that the solution is unique, we have to make an additional assumption in which we need another property of an ownership matrix. 

\begin{definition}
An ownership matrix $\bfM\in\R^{n\times n}$ possesses the \emph{Elsinger Property} if 
there exists no subset $\mathcal J \subset \mathcal N$ such that
\begin{equation}
\sum_{i\in\mathcal J} M_{ij} = 1 \quad \text{for all $j\in\mathcal J$}.
\end{equation}
\end{definition}
The name of this property is chosen because \citet{elsinger09} is, by the best knowledge of the authors, 
the first one to use this assumption in the context of ownership matrices and the valuation of systemic risk. 
For our model, we demand that the considered ownership matrices fulfill this property. 

\begin{assumption}\label{assu:holding_mat}
The Elsinger Property holds for the debt and the equity ownership matrices $\bfMd$ and $\bfMs$ .
\end{assumption}

Note that the fact that $\bfMd$ and $\bfMs$ are holding matrices is equivalent with the existence of
$(\bfI_n-\bfMd)^{-1}$ and $(\bfI_n-\bfMs)^{-1}$, as shown by \citet{elsinger09}.
Moreover, Theorem \ref{theo:unique_fp} below will show that Assumption \ref{assu:holding_mat} ensures that there is only one fixed point of $\Phi$. 
To show this, we introduce the two vectors 
\begin{equation}\label{eq:def_Rgreat}
\Rgreat = \begin{pmatrix} \rgreat \\ \sgreat \end{pmatrix} = \begin{pmatrix} \bfd \\ (\bfI_n-\bfMs)^{-1}(\bfa+\bfMd\bfd-\bfd)^+ \end{pmatrix}
\end{equation}
and 
\begin{equation}
\Rsmall = \begin{pmatrix} \rsmall \\ \ssmall \end{pmatrix} = \begin{pmatrix} \min\{\bfd, \bfa\}  \\ (\bfa-\bfd)^+ \end{pmatrix}
\end{equation}
The vector $\Rgreat$ assumes that the debt payments are fully recovered so that in the debt component $\bfr=\bfd$. 
Note that even if for a fixed point $\bfR^*=\left(\begin{smallmatrix}\bfr^* \\ \bfs^*\end{smallmatrix}\right)$ of $\Phi$ it holds that 
$\bfr^*=\bfd$, it must not necessarily hold that $\Rgreat=\bfR^*$. 
The second vector $\Rsmall$ emerges when the liquidation equations \eqref{eq:liq_eq_debt} and \eqref{eq:liq_eq_equity} are applied 
and the ownership structure of liabilities and equities is completely ignored. 
In this case the term $\bfMd\bfr + \bfMs\bfs$ that represents the income of each firm stemming from debt 
and equity cross-ownership is set to zero. 
Hence, the firms only have the exogenous assets $\bfa$ as an income. 
The vector $\Rsmall$ results from applying the mapping $\Phi$ in equation \eqref{eq:Phi} to the vector $\bzero_{2n}$, i.e.
\begin{equation}\label{eq:Rsmall_zero}
\Phi(\bzero_{2n}) = \Phi\begin{pmatrix}\bzero_n\\\bzero_n\end{pmatrix} 
= \begin{pmatrix} \min\{\bfd, \bfa\}  \\ (\bfa-\bfd)^+ \end{pmatrix} = \Rsmall.
\end{equation}

\begin{lemma}\label{lem:Phi_self-mapping}
With the definitions above, it holds that $\Phi([\Rsmall, \Rgreat])\subset[\Rsmall, \Rgreat]$.
\end{lemma}

\begin{proof}
Assume $\bfR\in[\Rsmall, \Rgreat]$. 
Because of \eqref{eq:Rsmall_zero} and the monotony of $\Phi$, we have that $\Phi(\bfR)\ge\Rsmall$. 
By definition of $\Phi$ and $\Rgreat$, only the lower $n$ lines of the vector inequality $\Phi(\bfR)\le\Rgreat$ need to be shown. 
By Lemma \ref{lem:inv_xos-mat} and \eqref{eq:def_Rgreat}, it holds that
\begin{equation}
\sgreat - \bfMs\sgreat = (\bfa + \bfMd\bfd - \bfd)^+
\end{equation}
and therefore
\begin{equation}
\sgreat \ge (\bfa + \bfMd\bfd + \bfMs\sgreat - \bfd)^+.
\end{equation}
\end{proof}

Before showing the importance of $\Rgreat$ and $\Rsmall$ as upper and lower bounds of the solution $\bfR^*$, 
we need to introduce the terms \emph{default set} and \emph{default matrix}. 
For $\bfr\ge\bzero_n$ and $\bfs\ge\bzero_n$ the set 
\begin{equation}\label{eq:defi_def_set}
D(\bfr,\bfs) = \left\{i\in\mathcal N:a_i + \sum_{j=1}^n M^{\bfd}_{ij} r_j + \sum_{j=1}^n M^{\bfs}_{ij} s_j < d_i\right\}
\end{equation}
is called \emph{default set under $\bfr$ and $\bfs$} because -- given $\bfr$ and $\bfs$ -- the firms in $D(\bfr,\bfs)$ 
are not able to fully satisfy their obligations and hence are in default. 
We say that firm $i$ is in default under $\bfr$ and $\bfs$ if $i\in D(\bfr,\bfs)$. 
For $\bfR=(\bfr^t,\bfs^t)^t$ we will sometimes abbreviate the default set as $D(\bfR)$.
The \emph{default matrix corresponding to $\bfr$ and  $\bfs$}, $\bfL(\bfr, \bfs)\in\R^{n\times n}$, is defined as
\begin{equation}\label{eq:defi_def_mat}
\bfL(\bfr,\bfs) = \diag(\bfa+\bfM^{\bfd}\bfr + \bfM^{\bfs}\bfs-\bfd < \bzero_n)
\end{equation}
and is the diagonal matrix with entry 1 for firms in default under $\bfr$ and $\bfs$ at the corresponding position 
and with the value 0 for firms not in default. 
With the new notation, we can show the crucial limiting property of $\Rgreat$ and $\Rsmall$. 

\begin{proposition}\label{prop:lim_R}
Let $\bfR^*$ be a non-negative solution of the fixed point problem defined in \eqref{eq:Phi}. 
Then $\bfR^*\in[\Rsmall,\Rgreat]$.
\end{proposition}

\begin{proof}
Because of \eqref{eq:Rsmall_zero} and the monotony of $\Phi$, $\bfR^*\ge \Rsmall$, so we only show the validity of the upper bound $\Rgreat$. 
Since $\bfR^*$ is a fixed point of $\Phi$, we can write
\begin{equation}
\Phi\begin{pmatrix}\bfr^* \\ \bfs^*\end{pmatrix} = 
\begin{pmatrix} \min\{\bfd, \bfa + \bfMd\bfr^* + \bfMs\bfs^*\} \\ (\bfa + \bfMd\bfr^* + \bfMs\bfs^* - \bfd)^+ \end{pmatrix} 
= \begin{pmatrix}\bfr^* \\ \bfs^*\end{pmatrix} = \bfR^*.
\end{equation}
Obviously, $\bfr^*\le\bfd=\rgreat$, hence we reduce our considerations to the equity component of $\bfR^*$ which, 
together with $\bfL(\bfr^*,\bfs^*)=\bfL^*$, can be presented as
\begin{equation}\label{eq:fp_eq_comp}
\bfs^* = (\bfa + \bfMd\bfr^* + \bfMs\bfs^* - \bfd)^+ = (\bfI_n-\bfL^*) (\bfa + \bfMd\bfr^* + \bfMs\bfs^* - \bfd). 
\end{equation}
Because of $(\bfI_n-\bfL^*)\bfs^* = \bfs^*$ we can reformulate \eqref{eq:fp_eq_comp} into
\begin{equation}
\begin{split}
\bfs^*  &= (\bfI_n-\bfL^*) \bfMs\bfs^* + (\bfI_n-\bfL^*) (\bfa + \bfMd\bfr^* - \bfd) \\
        &= (\bfI_n-\bfL^*) \bfMs(\bfI_n-\bfL^*)\bfs^* + (\bfI_n-\bfL^*) (\bfa + \bfMd\bfr^* - \bfd).
\end{split}		
\end{equation}
Rearranging yields to
\begin{equation}\label{eq:bfs_inv}
\bfs^* = (\bfI_n - (\bfI_n-\bfL^*)\bfMs(\bfI_n-\bfL^*))^{-1}(\bfI_n-\bfL^*)(\bfa + \bfMd\bfr^* - \bfd).
\end{equation}
Together with Lemma \ref{lem:inv_aux} in the Appendix, this leads to 
\begin{equation}
\begin{split}
\bfs^* &=   (\bfI_n - (\bfI_n-\bfL^*)\bfMs(\bfI_n-\bfL^*))^{-1}(\bfI_n-\bfL^*)(\bfa + \bfMd\bfr^* - \bfd) \\
       &\le (\bfI_n - (\bfI_n-\bfL^*)\bfMs(\bfI_n-\bfL^*))^{-1}(\bfI_n-\bfL^*)(\bfa + \bfMd\bfd - \bfd) \\
       &\le (\bfI_n - (\bfI_n-\bfL^*)\bfMs(\bfI_n-\bfL^*))^{-1}(\bfI_n-\bfL^*)(\bfa + \bfMd\bfd - \bfd)^+ \\
       &\le (\bfI_n-\bfL^*)(\bfI_n-\bfMs)^{-1}(\bfI_n-\bfL^*)(\bfa + \bfMd\bfd - \bfd)^+ \\
       &\le (\bfI_n-\bfMs)^{-1}(\bfa + \bfMd\bfd - \bfd)^+ \\
       &= \sgreat,
\end{split}		
\end{equation}
from which the assertion follows.
\end{proof}

Using the results from above, we can now show that there is only one fixed point which is in the interval $[\Rsmall,\Rgreat]$. 
The proof of the following theorem is given in the Appendix.

\begin{theorem}\label{theo:unique_fp}
Under Assumption \ref{assu:holding_mat} and for an arbitrary financial system $\F=(\bfa,\bfd,\bfMd,\bfMs)$, 
there exists a unique fixed point of the mapping $\Phi$. 
The fixed point $\bfR^*$ is non-negative and $\bfR^*\in[\Rsmall,\Rgreat]$.
\end{theorem}

In the sequel, we assume that Assumption \ref{assu:holding_mat} holds 
so that $\F$ has only one solution $\bfR^*\in(\R_0^+)^{2n}$, i.e.
\begin{equation}
\bfR^*=\begin{pmatrix}\bfr^*\\\bfs^*\end{pmatrix} = \Phi \begin{pmatrix}\bfr^*\\\bfs^*\end{pmatrix} = \Phi(\bfR^*).
\end{equation}
The requirements of Assumption \ref{assu:holding_mat} are less strict than the assumption that both $\|\bfMd\|<1$ and  $\|\bfMs\|<1$ 
that is used for example in \citet{fischer14}, \citet{suzuki02} or \citet{gourieroux12}. 
This follows by the fact that in case of $\bfMd$ having the Elsinger Property, it must not necessarily hold that $\|\bfMd\|<1$ 
and hence $\Phi$ is no strict contraction anymore. 
However, the assumption still guarantees that the solution of the system is unique.

\section{Non-Finite Algorithms}\label{sec:iterative_algo}

In this section, two existing solution algorithms that can be found in the literature are presented. 
One algorithm consists of the iterative use of the mapping $\Phi$ on a chosen starting vector and is given in the first subsection. 
A modification of this Picard Iteration is used in the work of \citet{elsinger09}, 
where for the determination of the equity component, a more sophisticated subalgorithm is used (Section \ref{subsec:elsinger}). 
In the last subsection, a new algorithm is developed that combines the ideas of \citet{elsinger09} and \citet{eisenberg01} 
which results in a faster convergence of the procedure. 

\subsection{The Picard Algorithm}

The most intuitive way to calculate $\bfR^*$ for the system $\F$ consists of the iterative use of $\Phi$. 
It will be shown in this section that with an arbitrary starting vector $\bfR^0\in(\R_0^+)^{2n}$, 
\begin{equation}\label{eq:picard-iteration}
\bfR^* = \lim_{l\to\infty}\Phi^l(\bfR^0) = \lim_{l\to\infty} \underbrace{\Phi\circ\ldots\circ\Phi}_{l}(\bfR^0),
\end{equation}
which is commonly known as the \emph{Picard Iteration}. 
Since $\bfR^*\ge\bzero$, the range for the starting vector $\bfR^0$ can be reduced to only non-negative vectors. 
Beyond that, the search for an optimal starting point can be limited to the interval $[\Rsmall, \Rgreat]$, as shown in Theorem \ref{theo:unique_fp}.
A direct consequence is that any iteration procedure that aims to calculate $\bfR^*$ 
should make sure that 
\begin{inparaitem}
  \item[(i)] no starting point of the iteration is chosen outside the interval $[\Rsmall,\Rgreat]$ and that
  \item[(ii)] every interim result of the procedure also needs to be in that interval. 
\end{inparaitem}
Otherwise, the procedure is inefficient. 
For these reasons, we present an algorithm that can start with both, $\Rgreat$ and $\Rsmall$. 

\begin{algorithm}{1}[Picard Algorithm]\label{alg:picard}\hfill
\begin{enumerate}
  \item For $k=0$, choose $\bfR^0\in[\Rsmall, \Rgreat]$ and $\eps > 0$.
  \item\label{alg:fp_picard}
  For $k\ge 1$, determine $\bfR^k = \Phi(\bfR^{k-1})$. 
  \item If $\|\bfR^{k-1}-\bfR^k\|<\eps$, stop the algorithm. 
  Else, set $k = k + 1$ and proceed with step \ref{alg:fp_picard}.
\end{enumerate}
\end{algorithm}

We will use the two expressions Picard Iteration and Picard Algorithm synonymously for Algorithm \ref{alg:picard}. 
In \citet{suzuki02} and \citet{fischer14} the Picard Iteration is the algorithm of choice to determine solutions of 
\eqref{eq:liq_eq_debt} and \eqref{eq:liq_eq_equity}. 

\begin{proposition}\label{prop:conv_picard}
In case of $\bfR^0=\Rsmall$, Algorithm \ref{alg:picard} generates a sequence of increasing vectors $\bfR^k$, 
and for $\bfR^0=\Rgreat$ a sequence of decreasing vectors. 
For all starting points, the algorithm converges to the solution $\bfR^*$. 
\end{proposition}

\begin{proof}
Let $\bfR^0=\Rsmall$, then
\begin{equation}
\Phi(\Rsmall)=\begin{pmatrix}\min\{\bfd, \bfa + \bfMd\rsmall + \bfMs\ssmall\} \\ (\bfa + \bfMd\rsmall + \bfMs\ssmall - \bfd)^+\end{pmatrix} 
\ge \begin{pmatrix} \min\{\bfd, \bfa\} \\ (\bfa - \bfd)^+ \end{pmatrix} = \Rsmall. 
\end{equation}
From the monotonicity of $\Phi$, it follows that for all iterates we have $\bfR^{k+1}\ge\bfR^k, k\ge 1$. 
For $\bfR^0=\Rgreat$, first check that because of $\sgreat=(\bfI_n-\bfMs)^{-1}(\bfa + \bfMd\bfd - \bfd)^+$ and $\rgreat = \bfd$,		
\begin{equation}
\begin{split}
& (\bfa + \bfMd\rgreat + \bfMs\sgreat - \bfd)^+ \\
& \quad = \left(\bfa + \bfMd\bfd - \bfd + \bfMs\sgreat - \sgreat + \sgreat\right)^+ \\
& \quad = \left(\bfa + \bfMd\bfd - \bfd - (\bfI_n-\bfMs)\sgreat + \sgreat\right)^+ \\
& \quad = \left(\underbrace{\bfa + \bfMd\bfd - \bfd - (\bfa + \bfMd\bfd - \bfd)^+}_{\le \bzero_n} +\; \sgreat\right)^+ \\
& \quad \le (\sgreat)^+ = \sgreat
\end{split}		
\end{equation}
and thus
\begin{equation}
\Phi(\Rgreat)= \begin{pmatrix}\min\{\bfd, \bfa + \bfMd\rgreat + \bfMs\rgreat\} \\ (\bfa + \bfMd\rgreat + \bfMs\sgreat - \bfd)^+\end{pmatrix} 
\le \begin{pmatrix} \bfd \\ \sgreat \end{pmatrix} = \Rgreat.
\end{equation}
Again it holds, due to the monotonicity of $\Phi$, that $\bfR^{k+1}\le\bfR^k, k\ge 1$.
Hence for any $\bfR\in[\Rsmall, \Rgreat]$ it follows because of $\bfR\le\Rgreat$ that $\Phi(\bfR)\le\Phi(\Rgreat)\le\Rgreat$ 
and with the same argumentation it follows that $\Phi(\bfR)\ge\Phi(\Rsmall)\ge\Rsmall$. 
This means that any series from the Picard Iteration with a starting point in the interval $[\Rsmall, \Rgreat]$ is bounded from above and from below. 
Since $\Phi$ is continuous, it follows that the series must converge to some $\widetilde \bfR$ 
such that $\Phi(\widetilde \bfR)=\widetilde \bfR$. 
According to Theorem \ref{theo:unique_fp}, there is only one fixed point, so it must hold that $\widetilde \bfR=\bfR^*$. 
\end{proof}

Regarding the Picard Iteration with the starting points $\Rsmall$ or $\Rgreat$, it should be mentioned that 
besides \citet{suzuki02} and \citet{fischer14}, also \citet{shin06} considers a Picard iteration in a
system valuation context. Shin's model is one with debt cross-ownership and multiple seniorities, while 
equity cross-ownership is not considered. As such, the model is situated somewhere between
the \citet{eisenberg01} and the \citet{elsinger09} framework. Instead of maturity values (as done here),
Shin directly considers risk-neutral values at time 0, and takes for the start of the iteration procedure
either a ``conservative'' viewpoint, where debt is assumed to have the value zero, or an ``optimistic''
viewpoint, where the value of debt is assumed to be the face value. Shin's starting
points therefore seem
to be risk-neutral time zero equivalents to the here considered $\Rsmall$ and $\Rgreat$.

The Picard Iteration -- or any other iterative algorithm in this section -- might not reach the solution $\bfR^*$ in finitely many iteration steps. 
Examples of financial systems with this property can easily be constructed. 
From a computational or practical point of view this means that iterative algorithms like the Picard Iteration have the disadvantage 
that under some circumstances many iterations are needed 
to approach to $\bfR^*$ sufficiently close, which makes these algorithms somewhat inefficient. 
The Trial-and-Error Algorithms presented in Section \ref{sec:def-set-algo} do not have this drawback 
since for these procedures it is assured that they will reach the solution in a finite number of steps.

\subsection{The \citeauthor{elsinger09} Algorithm}\label{subsec:elsinger}

In \citet{elsinger09}, an algorithm for $\bfR^*$ is presented which differs from the Picard Iteration. 
This procedure consists of splitting the two components of $\bfR$, the equity and the debt component, 
and apply different computation methods on both components in each iteration step. 
For the equity component, a sub-algorithm is applied where the equity payments of the system are determined 
assuming a fixed amount of debt payments. 
Denote this vector of debt payments in the following by $\bar\bfr$, hence $\bzero_n\le\bar\bfr\le\bfd$. 
Aim of the sub-algorithm is to find a fixed point of the mapping $\Phi^{\bfs}:(\R_0^+)^n\to(\R_0^+)^n$ with 
\begin{equation}\label{eq:phi_aux_eq}
\Phi^{\bfs}(\bfs; \bar\bfr) = (\bfa+\bfMd\bar\bfr + \bfM^{\bfs}\bfs - \bfd)^+.
\end{equation}
This mapping represents the equity component of $\Phi$ for a given debt payment of $\bar\bfr$. 
The fixed point of $\Phi^{\bfs}(\cdot;\bar\bfr)$ is denoted by $\bfs(\bar\bfr)$, i.e.
\begin{equation}\label{eq:phi_aux_eq_fp}
\Phi^{\bfs}(\bfs(\bar\bfr); \bar\bfr) = (\bfa+\bfMd\bar\bfr + \bfM^{\bfs}\bfs(\bar\bfr) - \bfd)^+ = \bfs(\bar\bfr).
\end{equation}
As shown in \citet{elsinger09}, this fixed point exists and is unique since $\bfMs$ has the Elsinger Property. 

The following algorithm delivers for given $\bar\bfr$ a series of vectors $\bfw^k\in\R^n$ that converge to a vector 
whose positive part is the fixed point of \eqref{eq:phi_aux_eq}. 
To explain this in more detail, first define for a given vector $\bfw\in\R^n$ the set
\begin{equation}
P(\bfw) = \left\{i\in\mathcal N: w_i \ge 0\right\}
\end{equation} 
and the matrix
\begin{equation}
\bfG(\bfw) = \diag(\bfw \ge \bzero_n)
\end{equation}
as the corresponding diagonal matrix. 
Note that these definitions of $P(\bfw)$ and $\bfG(\bfw)$ slightly differ from the original ones in \citet{elsinger09}, 
where a strictly larger sign was used. 
By our definition of default in \eqref{eq:defi_def_set}, a firm with zero equity value can still be not in default 
in the sense that all obligations can fully served. 
This situation is referred to as \emph{borderline firms} (cf. Section \ref{subsec:trial_error_alg_incr}).
However, this modification does not change the forthcoming theoretical results. 

\begin{algorithm}{2A}\label{alg:equity}\hfill
\begin{enumerate}
  \item For $k=0$, set $\bfw^0=\bfa+\bfMd\bar\bfr-\bfd$ and determine $P(\bfw^0)$ and $\bfG(\bfw^0)$.
  \item\label{alg:fp_equity} 
  For $k\ge 1$, solve $\Psi_{\bfw^k}(\bfw)=\bfw$ where
  \begin{equation}\label{eq:psi_s}
  \Psi_{\bfw^k}(\bfw) = \bfw^0 + \bfMs\bfG(\bfw^k)\bfw
  \end{equation}
  and denote the solution by $\bfw^{k+1}$, i.e. $\Psi_{\bfw^k}\left(\bfw^{k+1}\right)= \bfw^{k+1}$. 
  Determine $P(\bfw^{k+1})$ and $\bfG(\bfw^{k+1})$.
  \item If $P(\bfw^k)=P(\bfw^{k+1})$, stop the algorithm. 
  Else, set $k = k + 1$ and proceed with step \ref{alg:fp_equity}.
\end{enumerate}
\end{algorithm}

Before the properties of Algorithm \ref{alg:equity} are shown, we give some explanations for a better understanding of its functioning. 
The starting point is $\bfw^0$, which is the difference between $\bfa+\bfMd\bar\bfr$ and $\bfd$. 
The sum represents the firms incomes on their balance sheet that consists of the external assets and the payments due to cross-ownership of debt. 
Note that in this step the potential income from equity cross-ownership is ignored since $\bfMs$ does not appear. 
The idea is now as follows: The firms not in $P(\bfw^0)$ are not able to fully satisfy their liabilities (assuming debt payments of $\bar\bfr$) 
and will be in default. 
On the other hand, the firms that are in $P(\bfw^0)$ will be able to satisfy their obligees 
and can be regarded as solvent (again assuming debt payments of $\bar\bfr$), even though no intersystem payments due to equity cross-ownership 
are taken into account. 
As a consequence, the equity payments of the non-defaulting firms are added into the system via the product $\bfMs\bfG(\bfw^0)\bfw$. 
We can interpret the vector $\bfw^0$, as well as the other iterates $\bfw^k$, as pseudo equity vectors
that give us information about solvent and defaulting firms under the current debt and equity payments. 
The fact that the entries of $\bfw^k$ can be negative prevents that they can be naturally interpreted as equity vectors 
which is why we use the term ``pseudo''. 

The difference compared to the Picard Algorithm is that a linear equation system is solved to achieve a new equity payment vector 
instead of applying $\Phi$ to $(\bar\bfr^t, (\bfw^k)^t)^t$. 
This is because for the fixed point of $\Psi_{\bfw^k}$ it holds together with \eqref{eq:psi_s} that
\begin{equation}\label{eq:eq_alg_inv}
\bfw^{k+1} = (\bfI_n - \bfMs\bfG(\bfw^k))^{-1} \bfw^0
\end{equation}
Note that the inverse matrix exists since $\bfMs$ and hence $\bfMs\bfG(\bfs^k)$ have the Elsinger Property. 

The vector $\bfw^1$ can be interpreted as an ``updated'' version of $\bfw^0$ since the equity of the non-defaulting firms 
that are in $P(\bfw^0)$ is included in $\bfw^1$.
Based on the updated vector $\bfw^1$ it might appear that some firms that are not in $P(\bfw^0)$ have now non-negative entries in $\bfw^1$. 
This can be concluded from $\bfw^1\ge\bfw^0$ that we will show later. 
But these firms are now also able to contribute equity payments to the system. 
Consequently, the system has to be updated again by determining $\bfw^2$. 
The procedure continues until the set of defaulting firms stays the same from one iteration step to the next one. 

\begin{proposition}\label{prop:equ_comp_conv}
Given a fixed vector of debt payments $\bar\bfr\ge\bzero_n$:
\begin{enumerate}
  \item[(i)] Algorithm \ref{alg:equity} generates an increasing sequence of vectors $\bfw^k$. 
  \item[(ii)] Let $1\le l\le n$ such that
  \begin{equation}\label{eq:iter_eq_final}
  l:=\min\{j\in\{0, 1 \ld n\}:P(\bfw^j)=P(\bfw^{j+1})\}.
  \end{equation} 
  Then $\bfs(\bar\bfr)=(\bfw^{l+1})^+$ is the fixed point of the mapping $\Phi^{\bfs}(\cdot; \bar\bfr)$.  
  \item[(iii)] Let $d_0 = |P(\bfw^0)|\in\{0, 1 \ld n\}$ be the number of firms with a positive entry in $\bfw^0$.
  The fixed point $\bfs(\bar\bfr)$ is reached after no more than $n-d_0$ iteration steps.
\end{enumerate}
\end{proposition}

\begin{proof}
\begin{enumerate}
  \item[(i)] This part of the Proposition is shown by \citet{elsinger09}. 
  We give a different version of the proof.
  Because of \eqref{eq:eq_alg_inv}, the fact that $\bfG(\bfw^0)\bfw^0\ge\bzero_n$ 
  and using the series representation of $(\bfI_n - \bfMs\bfG(\bfw^0))^{-1}$ as shown in Lemma \ref{lem:inv_xos-mat} of the Appendix we get	
  \begin{equation}
  \begin{split}	
  \bfw^1 &= (\bfI_n - \bfMs\bfG(\bfw^0))^{-1} \bfw^0\\
         &= (\bfI_n + \bfMs\bfG(\bfw^0) + (\bfMs\bfG(\bfw^0))^2 + \ldots)\bfw^0 \\
         &= \bfw^0 + \bfMs\underbrace{\bfG(\bfw^0)\bfw^0}_{\ge\bzero_n} + 
            \bfMs\bfG(\bfw^0)\bfMs\underbrace{\bfG(\bfw^0)\bfw^0}_{\ge\bzero_n}+\ldots \\
         &\ge \bfw^0, 
  \end{split}		
  \end{equation}
  which is the induction start. 
  For the induction step we assume $\bfw^k\ge\bfw^{k-1}$ and $\bfG(\bfw^k)\ge\bfG(\bfw^{k-1})$ following from it. 
  We need to show that $\bfw^{k+1} \ge \bfw^k$, or, equivalently, $\bfw^{k+1} = \bfw^k + \bfe$ where $\bfe\ge \bzero_n$. 
  Since $\bfG(\bfw^k)\bfw^k\ge\bfG(\bfw^{k-1})\bfw^k$ and $\bfw^k=\bfw^0+\bfMs\bfG(\bfw^{k-1})\bfw^k$, it follows that
  \begin{equation}
  \bfu:= \bfw^0 + \bfMs\bfG(\bfw^k)\bfw^k - \bfw^k=\bfMs(\bfG(\bfw^k)-\bfG(\bfw^{k-1}))\bfw^k \ge \bzero_n.
  \end{equation}
  With this definition we have that 
  \begin{equation}
  \bfw^k+ \bfe = \bfw^0 + \bfMs\bfG(\bfw^k)(\bfw^k+ \bfe) = \bfw^0 + \bfMs\bfG(\bfw^k)\bfw^k +  \bfMs\bfG(\bfw^k)\bfe 
  \end{equation}
  and we can rearrange to 
  \begin{equation}
  \bfe - \bfMs\bfG(\bfw^k)\bfe = \bfw^0 + \bfMs\bfG(\bfw^k)\bfw^k - \bfw^k = \bfu \ge \bzero_n.
  \end{equation}
  Solving this for $\bfe$ leads to
  \begin{equation}
  \bfe = (\bfI_n - \bfMs\bfG(\bfw^k))^{-1}\bfu \ge \bzero_n
  \end{equation}
  from which follows that $\bfw^{k+1}\ge\bfw^k$.
  
  \item[(ii)]
  First, we will show that once a ``stable system'' has been reached, i.e. for $k\ge0$ we have $P(\bfw^k)=P(\bfw^{k+1})$,
  the sequence $\bfw^k$ will be constant.  
  Let $l$ be defined as above in \eqref{eq:iter_eq_final}. 
  Note that such an $l$ exists since $\bfw^k\le\bfw^{k+1}$ and therefore $P(\bfw^{k+1})\supseteq P(\bfw^k)$ for all $k\ge0$ as shown above. 
  Due to $\bfG(\bfw^l) = \bfG(\bfw^{l+1})$, it follows because of
  \begin{equation}
  \Psi_{\bfw^l}(\bfw) = \bfw^0 + \bfMs\bfG(\bfw^l)\bfw = \bfw^0 + \bfMs\bfG(\bfw^{l+1})\bfw = \Psi_{\bfw^{l+1}}(\bfw)
  \end{equation}
  that the two mappings $\Psi_{\bfw^l}$ and $\Psi_{\bfw^{l+1}}$ are the same and consequently $\bfw^{l+1}=\bfw^{l+2}$. 
  A direct consequence is that $P(\bfw^{l+2}) = P(\bfw^{l+1}) = P(\bfw^l)$ which implies $\bfG(\bfw^{l+2})=\bfG(\bfw^{l+1})=\bfG(\bfw^l)$. 
  By iteration, all following vectors will be equal to $\bfw^{l+1}$. 
  
  What remains to be shown out is that the positive part of this iteration vector is the fixed point of the mapping $\Phi^{\bfs}(\cdot; \bar\bfr)$. 
  Since $\bfw^{l+1}$ is the fixed point of $\Psi_{\bfw^l}$ it holds that $\bfw^{l+1} = \bfw^0 + \bfMs\bfG(\bfw^l)\bfw^{l+1}$. 
  This yields to 
  \begin{equation}\label{eq:fp_prop_w}
  \begin{split}
  \Phi^{\bfs}((\bfw^{l+1})^+; \bar\bfr) &= (\bfa + \bfMd\bar\bfr + \bfMs(\bfw^{l+1})^+ - \bfd)^+ \\
                                        &= (\bfa + \bfMd\bar\bfr + \bfMs\bfG(\bfw^{l+1})\bfw^{l+1} - \bfd)^+ \\
                                        &= (\bfa + \bfMd\bar\bfr + \bfMs\bfG(\bfw^l)\bfw^{l+1} - \bfd)^+ \\
                                        &= (\bfw^0 + \bfMs\bfG(\bfw^l)\bfw^{l+1})^+ \\
                                        &= (\bfw^{l+1})^+.
  \end{split}                                      
  \end{equation}
  \item[(iii)]
  As shown, the series $\bfw^k$ increases which means that the firms in $P(\bfw^0)$ will maintain 
  their positive entries in every further iteration step. 
  The same statement holds for every firm $i$ with $w_i^k < 0$ and $w_i^{k+1} \ge 0$ for any $k\ge 0$. 
  Because of (ii) this means that the number of iteration steps would certainly be maximal, 
  if in every iteration step the set $P(\bfw^k)$ increased by one and if $|P(\bfw^{l+1})|=n$. 
  In that case we would therefore have $|P(\bfw^{l+1})|-|P(\bfw^0)| = n-d_0$ maximal possible iteration steps.
\end{enumerate}
\end{proof}

Using Algorithm \ref{alg:equity} to get an equity vector for a given debt payment vector,
we can now present the algorithm to calculate the solution $\bfR^*$. 
In the sequel, we will make use of the mapping $\Phi^{\bfd}:(\R_0^+)^n\to(\R_0^+)^n$ defined by
\begin{equation}\label{eq:phi_aux_debt}
\Phi^{\bfd}(\bfr; \bar\bfs) = \min\{\bfd, \bfa+\bfMd\bfr+\bfMs\bar\bfs\}
\end{equation}
that represents the debt component of $\Phi$ for a given equity payment vector $\bar\bfs\ge\bzero_n$. 

\begin{algorithm}{3}[Elsinger Algorithm]\label{alg:elsinger}
Set $\eps>0$.
\begin{enumerate}
  \item For $k=0$, choose $\bfr^0\in\{\rsmall, \rgreat\}$ and determine $\bfs(\bfr^0)$ using Algorithm \ref{alg:equity}.
  \item\label{alg:fp_elsinger}
  For $k\ge1$, set $\bfr^k = \Phi^{\bfd}(\bfr^{k-1}; \bfs(\bfr^{k-1}))$ and calculate $\bfs(\bfr^k)$ by Algorithm \ref{alg:equity}.
  \item If $\left\|\begin{pmatrix}\bfr^{k-1}\\\bfs^{k-1}\end{pmatrix} - \begin{pmatrix}\bfr^k\\\bfs^k\end{pmatrix}\right\|<\eps$, 
  stop the algorithm. 
  Else, set $k = k + 1$ and proceed with step \ref{alg:fp_elsinger}.
\end{enumerate}
\end{algorithm}

The algorithm starts either assuming that all firms can fully deliver on their debt obligations ($\bfr^0=\rgreat=\bfd$) 
or that all firms have only their exogenous assets for paying their obligations ($\bfr^0=\rsmall=\min\{\bfd, \bfa\}$).
With this payment vector, the corresponding equity payments are obtained by using Algorithm \ref{alg:equity}. 
In the next step the debt vector has to be adapted to the new equity payments which is done applying $\Phi^{\bfd}$ to the previous debt vector. 
The updated debt payment vector is then used for determining a new equity payment vector. 
This procedure continues until the iterates are sufficiently close to each other. 
Additional to the original algorithm first presented in \citet{elsinger09}, 
Algorithm \ref{alg:elsinger} contains the second possible starting point $\rsmall$. 
We will show in the next proposition that if $\bfr^0=\rsmall$ is chosen, the vector of debt and equity payments establish an increasing sequence 
and hence converges to the solution $\bfR^*$ from below, while for $\bfr^0=\rgreat$, it converges from above.

\begin{proposition}\label{prop:elsinger_conv}
The Elsinger Algorithm delivers a series of decreasing vectors if $\bfr^0=\rgreat$ and a series of increasing vectors if $\bfr^0=\rsmall$. 
Both series converge to the fixed point of the mapping $\Phi$ in \eqref{eq:Phi}. 
\end{proposition}

\begin{proof}
The decreasing part is shown in \citet{elsinger09}, we only have to show that the debt iterate in the algorithm therein 
is identical to $\bfr^k$ in Algorithm \ref{alg:elsinger}. 
With our notation, the iterate of the debt component in \citet{elsinger09} is defined as 
\begin{equation}\label{eq:phi_elsinger}
\bfr^k = \min\{\bfd, (\bfw^*(\bfr^{k-1}) + \bfd)^+\},
\end{equation}
where $\bfw^*(\bfr^{k-1})$ is the solution of 
\begin{equation}\label{eq:equity_fp_els}
\bfw = \bfa + \bfMd\bfr^{k-1} + \bfMs\bfw^+ - \bfd.
\end{equation}
However, it follows from \eqref{eq:fp_prop_w} that for $\bar\bfr=\bfr^{k-1}$, 
\begin{equation}
(\bfw^{l+1})^+ = \Phi^{\bfs}\left((\bfw^{l+1})^+;\bfr^{k-1}\right) = \left(\bfa + \bfMd\bfr^{k-1} + \bfMs(\bfw^{l+1})^+ - \bfd\right)^+ = \left(\bfw^*(\bfr^{k-1})\right)^+, 
\end{equation}
where $\bfw^{l+1}$ is the result of Algorithm \ref{alg:equity} with the debt payment vector $\bfr^{k-1}$, 
i.e. $(\bfw^*(\bfr^{k-1}))^+ = \bfs(\bfr^{k-1})$. 
Because of \eqref{eq:equity_fp_els}, we have that
\begin{equation} 
\bfw^*(\bfr^{k-1}) = \bfa + \bfMd\bfr^{k-1} - \bfd + \bfMs(\bfw^*(\bfr^{k-1}))^+,
\end{equation}
from which follows with \eqref{eq:phi_aux_debt} and $\bfa\ge\bzero_n$ that
\begin{equation}
\begin{split}
\bfr^k &= \min\{\bfd, (\bfw^*(\bfr^{k-1}) + \bfd)^+\} \\
       &= \min\{\bfd, \bfa + \bfMd\bfr^{k-1} + \bfMs(\bfw^*(\bfr^{k-1}))^+\} \\
       &= \Phi^{\bfd}(\bfr^{k-1}; (\bfw^*(\bfr^{k-1}))^+)\\
       &= \Phi^{\bfd}(\bfr^{k-1}; \bfs(\bfr^{k-1})).
\end{split}       
\end{equation} 

What remains to be shown is that for the starting point $\bfr^0=\rsmall$ the generated series increases 
and converges to $\bfR^*$, which is done by induction. 
For the induction start check that 
\begin{equation}
\bfr^0 = \min\{\bfd, \bfa\} \le \min\{\bfd, \bfa + \bfMd\bfr^0 + \bfMs\bfs(\bfr^0)\} = \Phi^{\bfd}(\bfr^0;\bfs(\bfr^0)) = \bfr^1.
\end{equation}
As shown in \citet{elsinger09}, the result $\bfw^*(\bfr)$ of Algorithm \ref{alg:equity} is increasing in $\bfr$ from which also follows that 
$\bfs(\bfr)$ is increasing in $\bfr$. 
Hence, $\bfs(\bfr^0)\le\bfs(\bfr^1)$ which completes the induction start.
Assume for the induction step that $\bfr^{k-1}\le \bfr^k$ and consequently $\bfs(\bfr^{k-1})\le \bfs(\bfr^k)$. 
The next debt iterate emerges as		
\begin{equation}
\begin{split}
\bfr^{k+1} &= \min\{\bfd,\bfa + \bfMd\bfr^k + \bfMs(\bfs(\bfr^k))^+\} \\
           &\ge \min\{\bfd,\bfa + \bfMd\bfr^{k-1} + \bfMs(\bfs(\bfr^{k-1}))^+\} \\
           &= \bfr^k,
\end{split}		
\end{equation}
from which also follows that $\bfs(\bfr^{k+1})\ge\bfs(\bfr^k)$ and, hence, the increasing property of the series. 
For the convergence, check that $\bfs(\bfr^k)\ge\bzero_n$ and it holds that
\begin{equation}
\bfs(\bfr^k) = (\bfa + \bfMd\bfr^k + \bfMs\bfs(\bfr^k) - \bfd)^+ \le \bfMs\bfs(\bfr^k) + (\bfa + \bfMd\bfr^k - \bfd)^+.
\end{equation}
Because of $\bfr^k\le\bfr^*$, it follows after some rearrangements that
\begin{equation}\label{eq:els_algo_start_bound}
\bfs(\bfr^k) \le (\bfI_n - \bfMs)^{-1} (\bfa + \bfMd\bfr^k - \bfd)^+ \le (\bfI_n - \bfMs)^{-1} (\bfa + \bfMd\bfr^* - \bfd)^+,
\end{equation}
hence the series $\bfs(\bfr^k)$ is bounded from above as well and therefore converges to some $\mathbf{s}^*$ from below.
The fact that $\Phi^{\bfs}$ is continuous in $(\bfr^t,\bfs^t)^t$ implies together with $\Phi^{\bfs}(\bfs(\bfr^k);\bfr^k)=\bfs(\bfr^k)$ 
that $\Phi^{\bfs}(\bfs^*;\bfr^*)=\bfs^*$. 
Thus $((\bfr^*)^t,(\bfs^*)^t)^t$ solves \eqref{eq:liq_eq_equity}. 
Similarly, we can argue that because of the continuity of $\Phi^{\bfd}$, $\Phi^{\bfd}(\bfr^*;\bfs^*)=\bfr^*$ 
from which follows that $((\bfr^*)^t,(\bfs^*)^t)^t$ also solves \eqref{eq:liq_eq_debt} and therefore must be the fixed point $\bfR^*$.
\end{proof}

As described above, the Elsinger Algorithm determines the equity component of the iterates $\bfR^k$ in a different way than the Picard Iteration. 
An important consequence of this approach is that the iterates of the Elsinger Algorithm will for the decreasing version be in every step 
smaller than the iterates of the Picard Algorithm, as we will show in the next proposition. 
The same statement holds for the increasing version of both procedures, where the iterates from the Elsinger Algorithm will be greater 
than the iterates form the Picard Algorithm. 
Both procedures are difficult to compare concerning their total calculation effort due to different ways of obtaining the next 
equity iterate (cf. Section \ref{subsec:algo_efficency}). 
However, if we only take the number of needed iterations as a quality criterion, 
we can conclude that the Elsinger Algorithm converges faster to $\bfR^*$ than the Picard Iteration, 
no matter whether the algorithms start from the upper or the lower boundary. 

\begin{proposition}\label{prop:elsinger_better_picard}
Let $\bfR^k_{\rm P} = ((\bfr^k_{\rm P})^t, (\bfs^k_{\rm P})^t)^t$ be the $k$-th iterate of the Picard Algorithm 
and  $\bfR^k_{\rm E} = ((\bfr^k_{\rm E})^t, (\bfs^k_{\rm E})^t)^t$ the corresponding iterate of the Elsinger Algorithm. 
\begin{enumerate}[(i)]
  \item For any iterate $k\ge 0$ it holds that $\bfR^k_{\rm P}\ge\bfR^k_{\rm E}$ if $\bfR^0_{\rm P} = \Rgreat$ 
  and $\bfR^0_{\rm E} =(\rgreat^t, \bfs(\rgreat)^t)^t$. 
  In case of $\bfR^0_{\rm P} = \Rsmall$ and $\bfR^0_{\rm E} =(\rsmall^t, \bfs(\rsmall)^t)^t$, 
  we have that $\bfR^k_{\rm P}\le\bfR^k_{\rm E}$ for every iterate. 
  \item Let $\bfR^k$, $k\ge1$, be an iterate either of the Picard Algorithm with $\bfR^0=\Rgreat$ 
  or of the Elsinger Algorithm with $\bfR^0 = (\rgreat^t, \bfs(\rgreat)^t)^t$. 
  Then $\bfR^{k+1}_{\rm P}(\bfR^k)\ge\bfR^{k+1}_{\rm E}(\bfR^k)$. 
  If the starting vector is either $\bfR^0=\Rsmall$ or $\bfR^0 = (\rsmall^t, \bfs(\rsmall)^t)^t$, 
  it holds that $\bfR^{k+1}_{\rm P}(\bfR^k) \le \bfR^{k+1}_{\rm E}(\bfR^k)$.
\end{enumerate}
\end{proposition}

\begin{proof}
\begin{enumerate}[(i)]
  \item The assertion is shown by induction. 
  For $k=0$, suppose that the upper boundary is the starting vector for both algorithms. 
  In Equation \eqref{eq:els_algo_start_bound} it was shown that 
  \begin{equation}
  \bfs^0_{\rm E} = \bfs(\bfd) \le (\bfI_n-\bfMs)^{-1}(\bfa + \bfMd\bfd - \bfd)^+ = \sgreat = \bfs^0_{\rm P}. 
  \end{equation}
  Since $\bfr^0_{\rm E}=\bfr^0_{\rm P}=\bfd$, the induction start is complete. 
  Assume now, that for $k\ge1$ it holds that $\bfR^k_{\rm P}\ge\bfR^k_{\rm E}$. 
  From Proposition \ref{prop:elsinger_conv}, we know that $\bfR^{k+1}_{\rm E}\le\bfR^k_{\rm E}$. 
  This leads to 
  \begin{equation}
  \bfr^{k+1}_{\rm P} = \min\{\bfd, \bfa + \bfMd\bfr^k_{\rm P} + \bfMs\bfs^k_{\rm P}\} \ge 
  \min\{\bfd, \bfa + \bfMd\bfr^k_{\rm E} + \bfMs\bfs^k_{\rm E}\} = \bfr^{k+1}_{\rm E}
  \end{equation}
  and 		
  \begin{equation}
  \begin{split}
  \bfs^{k+1}_{\rm P} &= (\bfa + \bfMd\bfr^k_{\rm P} + \bfMs\bfs^k_{\rm P} - \bfd)^+\\
   &\ge (\bfa + \bfMd\bfr^k_{\rm E} + \bfMs\bfs^k_{\rm E} - \bfd)^+ \\
   &\ge (\bfa + \bfMd\bfr^{k+1}_{\rm E} + \bfMs\bfs^{k+1}_{\rm E} - \bfd)^+ \\
   &= \bfs^{k+1}_{\rm E}.
  \end{split}
  \end{equation}
  If the starting vector is the lower boundary and the series $\bfR^k_{\rm P}$ and $\bfR^k_{\rm E}$ are increasing, 
  the argumentation is similar. 
  \item We prove the claim for the decreasing version of the algorithms, the proof for the reverse direction is similar. 
  First, let $\bfR^k=\bfR^k_{\rm P}$. 
  The next iteration of the debt component is equal for both algorithms, 
  i.e. $\bfr^{k+1}_{\rm P} = \Phi^{\bfd}(\bfr^k;\bfs^k) = \bfr^{k+1}_{\rm E}$. 
  For the equity component, it holds that $\bfs^{k+1}_{\rm P}=\Phi^{\bfs}(\bfs^k;\bfr^k)$. 
  The mapping $\Phi^{\bfs}(\cdot;\bfr^k)$ has a unique fixed point, that we denote by $\bfs(\bfr^k)$ and that can be obtained via a Picard Iteration:
  \begin{equation}
  \lim_{l\to\infty}\left(\Phi^{\bfs}\right)^l(\bfs^k;\bfr^k) = \bfs(\bfr^k). 
  \end{equation}
  The iterates obviously form a decreasing sequence so that
  \begin{equation}
  \bfs^{k+1}_{\rm P} \ge \bfs(\bfr^k) \ge \bfs(\bfr^{k+1}) = \bfs^{k+1}_{\rm E}, 
  \end{equation}
  where the second inequality follows from the fact that $\bfs(\bfr)$ is increasing in $\bfr$ (cf. \citet{elsinger09}). 
  If the $k$-th iterate is given by $\bfR^k=\bfR^k_{\rm E}$, the arguments are analogous to the ones above. 
\end{enumerate}
\end{proof}

\subsection{A Hybrid Algorithm}

To motivate the approach of the next algorithm, we have to compare the functioning of the Elsinger Algorithm and the Picard Iteration. 
The major difference between both iterations emerges in the calculation of the equity component. 
Suppose that we are in iteration step $k\ge0$ and want to calculate the next iteration of the equity component. 
We ignore for an instant that both algorithms deliver different iterates 
and assume that the $k$-th iterate is given by $\bfR^k=((\bfr^k)^t, (\bfs^k)^t)^t$. 
In the Elsinger Algorithm, $\bfr^{k+1}$ is calculated first and then $\bfs^{k+1}$ as the fixed point of $\Phi^{\bfs}(\cdot;\bfr^{k+1})$ 
so that it holds that $\bfs^{k+1}=\Phi^{\bfs}(\bfs^{k+1};\bfr^{k+1})$. 
The Picard iterate, on the other side, can be written as $\bfs^{k+1}=\Phi^{\bfs}(\bfs^k;\bfr^k)$ from which it becomes clear 
that the Picard Iteration neither uses the ``updated'' debt vector $\bfr^{k+1}$, 
nor does it solve a separate fixed point mapping to obtain $\bfs^{k+1}$.

The determination of the debt component $\bfr^{k+1}$, however, is comparable in both algorithms. 
Again starting with $\bfR^k$ we have that $\bfr^{k+1}=\Phi^{\bfd}(\bfr^k;\bfs^k)$ for both procedures. 
An obvious extension of the Elsinger Algorithm would be to utilize the principle used for the equity component for the debt component as well. 
In the article of \citet{eisenberg01}, this concept is used for systems with no cross-ownership of equity, 
i.e. where $\bfMs=\bzero_{n\times n}$. 
In this subsection we will generalize the results of this work and it will turn out that combining both ideas, the one of \citet{elsinger09} 
and the one of \citet{eisenberg01}, will help to minimize the number of needed iteration steps of the global algorithm to find $\bfR^*$.

To explain this idea in more detail, say that for a debt payment vector $\bar\bfr\in[\rsmall, \rgreat]$ we have a corresponding 
equity vector $\bfs(\bar\bfr)$, that is, a fixed point of the mapping $\Phi^{\bfs}(\cdot; \bar\bfr)$ in \eqref{eq:phi_aux_eq}. 
In the Elsinger Algorithm the next debt iterate emerges as $\Phi^{\bfd}(\bar\bfr; \bfs(\bar\bfr))$. 
Instead of using this iterate, our aim is now to find the fixed point of $\Phi^{\bfd}(\cdot; \bfs(\bar\bfr))$ as the new iterate. 
This can be done using the following Algorithm.

\begin{algorithm}{4A}\label{alg:debt}
Suppose $\bfs\ge\bzero_n$. 
\begin{enumerate}
  \item For $k=0$, set $\bfr^0=\bar\bfr$ and determine $D(\bfr^0,\bfs)$ and $\bfL(\bfr^0,\bfs)$. 
  \item\label{alg:fp_debt}
  For $k\ge 1$, solve $\Theta_{\bfr^{k-1},\bfs}(\bfr) = \bfr$ where
  \begin{equation}\label{eq:fp_debt}
  \begin{split}
  \Theta_{\bfr^{k-1},\bfs}(\bfr) &= \bfL(\bfr^{k-1},\bfs)\left(\bfa + \bfMd\left(\bfL(\bfr^{k-1},\bfs)\bfr + 
  \left(\bfI_n-\bfL(\bfr^{k-1},\bfs)\right)\bfd\right) + \bfMs\bfs \right) \\ 
  &\qquad + \left(\bfI_n - \bfL(\bfr^{k-1},\bfs)\right)\bfd
  \end{split}
  \end{equation}
  \item Denote the solution by $\bfr^k$, i.e. $\Theta_{\bfr^{k-1},\bfs}(\bfr^k)=\bfr^k$ 
  and determine $D(\bfr^k,\bfs)$ and $\bfL(\bfr^k,\bfs)$. 
  \item If $D(\bfr^k,\bfs) = D(\bfr^{k-1},\bfs)$, stop the algorithm. 
  Else, set $k = k + 1$ and proceed with step \ref{alg:fp_debt}.
\end{enumerate}
\end{algorithm}

The algorithm is identical to the one given in \citet{eisenberg01} with the modification 
that some additional fixed payments due to equity cross-ownership are included. 
It solves \eqref{eq:liq_eq_debt} for a fixed amount of equity payment $\bfs\ge\bzero_n$, 
i.e. is the fixed point of the mapping $\Phi^{\bfd}(\cdot, \bfs)$, as we will show in the next proposition. 
Denote this fixed point by $\bfr^*(\bfs)$ for instance.
In the Hybrid Algorithm following later on, $\bfr^*(\bfs)$ is used as the next iterate for the debt component. 
To see the difference between the calculation of the debt component in the Elsinger Algorithm, 
assume that an arbitrary debt payment vector $\bfr\in[\bzero_n,\bfd]$ is given 
and that the corresponding equity payment vector $\bfs(\bfr)$ is given too. 
The fixed point of the mapping $\Phi^{\bfd}(\cdot;\bfs(\bfr))$ can on the one hand be obtained using Algorithm \ref{alg:debt} above, 
but on the other hand, we could also use a Picard Iteration, since for any $\bfr\in[\bzero_n,\bfd]$ it holds that
\begin{equation}
\bzero_n \le \Phi^{\bfd}(\bfr;\bfs(\bfd)) = \min\{\bfd, \bfa + \bfMd\bfr + \bfMs\bfs(\bfd)\} \le \bfd.
\end{equation}
Starting with the vector $\bfr$, the fixed point $\bfr^*(\bfs)$ is given as
\begin{equation}\label{eq:picard_debt}
\bfr^*(\bfs) = \lim_{l \to \infty}(\Phi^{\bfd})^l(\bfr;\bfs(\bfd)).
\end{equation}
In the Elsinger Algorithm, however, the next iterate for the debt component is defined as $\Phi^{\bfd}(\bfr;\bfs(\bfr))$ 
which is therefore the first iterate of the Picard Iteration in \eqref{eq:picard_debt}. 
Hence, one can say that using in the Elsinger Algorithm, a simple mapping is applied to obtain the next iterate, 
whereas in the Hybrid Algorithm, the fixed point of a mapping is determined. 
In Proposition \ref{prop:hybrid_better_elsinger}, we will show that when using the idea of the latter algorithm, 
the iterates $\bfR^k$ will always be closer to the searched solution $\bfR^*$.

\begin{proposition}\label{prop:debt_comp_conv}
Let $\bar\bfr\in[\rsmall, \rgreat]$ be a debt payment vector and $\bfs\ge\bzero_n$ a vector of equity payments such that 
\begin{equation}\label{eq:assu_debt_decr}
\Phi^{\bfd}(\bar\bfr; \bfs) \le \bar\bfr.
\end{equation}
\begin{enumerate}
  \item[(i)] Algorithm \ref{alg:debt} generates a well-defined decreasing sequence of vectors $\bfr^k$. 
  \item[(ii)] Let $1\le l\le n$ such that
  \begin{equation}\label{eq:iter_debt_final}
  l:=\min\{j\in\{0, 1 \ld n\}:D(\bfr^j, \bfs)=D(\bfr^{j+1}, \bfs)\}.
  \end{equation} 
  Then $\bfr^*(\bfs)=\bfr^{l+1}$ is the fixed point of the mapping $\Phi^{\bfd}(\cdot; \bfs)$ defined in \eqref{eq:phi_aux_debt}.
  \item[(iii)] Let $d_0=|D(\bar\bfr, \bfs)|$ be the number of firms in default under $\bar\bfr$ and $\bfs$. 
  The fixed point $\bfr^*(\bfs)$ is reached after no more than $n-d_0$ iteration steps.
\end{enumerate}
\end{proposition}

\begin{proof}
Since the equity vector $\bfs$ is considered as fixed we can modify the financial system $\F$ 
by setting $\tilde\bfa=\bfa+\bfMs\bfs$ and $\widetilde\bfM^{\bfs}=\bzero_{n\times n}$. 
The new system $\tilde\F =(\tilde\bfa,\bfd,\bfMd,\widetilde\bfM^{\bfs})$ is then a system without cross-ownership of equity. 
Such systems are considered in \citet{eisenberg01}. 
\begin{enumerate}
  \item[(i)] The proof that the sequence $\bfr^k$ decreases is now equivalent to the proof given in \citet{eisenberg01}. 
  A needed assumption in the proof therein is that $\bar\bfr$ is a so-called supersolution 
  which is given because of $\Phi^{\bfd}(\bar\bfr; \bfs(\bar\bfr)) \le \bar\bfr\le\bfd$. 
  What we have to show to complete this part is that the fixed point of the mapping in \eqref{eq:fp_debt} exists and is unique, 
  since their definition of a financial system, differs slightly from ours.  
  Denote by $\bfL:=\bfL(\bfr^k,\bfs)$ the diagonal matrix for $\bfr^k$. 
  The next iterate $\bfr^{k+1}$ is according to \eqref{eq:fp_debt} given by
  \begin{equation}
  \begin{split}
  \bfr^{k+1} &= \bfL\left(\tilde\bfa + \bfMd(\bfL\bfr^{k+1} + (\bfI_n-\bfL)\bfd) \right) + (\bfI_n-\bfL)\bfd \\
             &= \bfL\bfMd\bfL\bfr^{k+1} + \bfL\left(\tilde\bfa + \bfMd(\bfI_n-\bfL)\bfd\right) + (\bfI_n-\bfL)\bfd
  \end{split}
  \end{equation}
  and rearranging yields to 
  \begin{equation}
  \bfr^{k+1} = \left(\bfI_n-\bfL\bfMd\bfL\right)^{-1}\left(\bfL\left(\tilde\bfa + \bfMd(\bfI_n-\bfL)\bfd\right) + (\bfI_n-\bfL)\bfd\right).
  \end{equation}
  Note that $\bfMd$ has the Elsinger Property and, hence, so does $\bfL\bfMd\bfL$, which means that the inverse of $\bfI_n-\bfL\bfMd\bfL$ exists. 
  This proves the uniqueness of $\bfr^{k+1}$.
  
  \item[(ii)]  The argumentation that the sequence converges and becomes constant in the end is analogous to part (ii) 
  of the proof of Proposition \ref{prop:equ_comp_conv}. 
  Since $\bfr^k$ is decreasing, we have that $D(\bfr^k,\bfs)\subseteq D(\bfr^{k+1},\bfs)$ 
  that means the number of firms in default increases. 
  If $D(\bfr^l, \bfs)=D(\bfr^{l+1}, \bfs)$, then we also have that $\bfL(\bfr^l,\bfs)= \bfL(\bfr^{l+1},\bfs)$, 
  from which follows that the mappings $\Theta_{\bfr^l,\bfs}$ and $\Theta_{\bfr^{l+1},\bfs}$ have the same fixed point. 
  It must hold then that all consequent iterates are equal.
  
  To show that $\bfr^{l+1}$ is the fixed point of $\Phi^{\bfd}(\cdot,\bfs)$, first check that 
  by definition of $\bfL(\bfr^{l+1},\bfs)$:
  \begin{equation}
  \bfr^{l+1} = \bfL(\bfr^{l+1},\bfs)\bfr^{l+1} + (\bfI_n - \bfL(\bfr^{l+1},\bfs))\bfd.
  \end{equation}
  It then holds that
  \begin{equation}
  \begin{split}
  \Phi^{\bfd}(\bfr^{l+1}; \bfs) &= \min\{\bfd, \tilde\bfa + \bfMd\bfr^{l+1}\} \\
   &= (\bfI_n - \bfL(\bfr^{l+1},\bfs))\bfd + \bfL(\bfr^{l+1},\bfs)(\tilde\bfa + \bfMd\bfr^{l+1}) \\
   &= (\bfI_n - \bfL(\bfr^l,\bfs))\bfd \\
   &\qquad + \bfL(\bfr^l,\bfs)\left(\tilde\bfa + \bfMd\left(\bfL(\bfr^{l+1},\bfs)\bfr^{l+1} + (\bfI_n - \bfL(\bfr^{l+1},\bfs))\bfd\right) \right)\\ 
   &= (\bfI_n - \bfL(\bfr^l,\bfs))\bfd  + \bfL(\bfr^l,\bfs)\left(\tilde\bfa + \bfMd\left(\bfL(\bfr^l,\bfs)\bfr^{l+1} + (\bfI_n - \bfL(\bfr^l,\bfs))\bfd\right) \right)\\ 
   &= \bfr^{l+1},
  \end{split} 
  \end{equation}
  where the last equality follows from \eqref{eq:fp_debt}.
  \item[(iii)] This part is similar to part (iii) of the proof of Proposition \ref{prop:equ_comp_conv} with the reverse argumentation. 
  The $d_0$ firms in default under the starting vector will stay in default since the series decreases. 
  To achieve a maximum theoretical length of the algorithm, exactly one additional default step has to occur in every new iteration step. 
  This results in no more than $n-d_0$ possible iteration steps.
\end{enumerate}
\end{proof}

The validity of the inequality in \eqref{eq:assu_debt_decr} is crucial for the monotonicity 
of the iterates $\bfr^k$ produced by Algorithm \ref{alg:debt}. 
However, there are situations in which a debt payment vector $\bar\bfr\in[\rsmall,\rgreat]$ is given 
together with an arbitrary vector $\bfs\ge\bzero_n$ and where \eqref{eq:assu_debt_decr} does not hold. 
Think of an algorithm to find $\bfR^*$ that starts with $\Rsmall$. 
In this case the first debt iterate is $\rsmall=\min\{\bfd,\bfa\}$ and the corresponding equity iterate is $\bfs(\rsmall)$. 
Applying $\Phi^{\bfd}$ on these vectors yields to
\begin{equation}
\Phi^{\bfd}(\rsmall;\bfs(\rsmall)) = \min\{\bfd, \bfa + \bfMd\rsmall + \bfMs\bfs(\rsmall)\} \ge \min\{\bfd,\bfa\} = \rsmall
\end{equation}
and to a violation of \eqref{eq:assu_debt_decr}. 
Finding the next debt iterate as the fixed point of $\Phi^{\bfd}(\cdot;\bfs(\rsmall))$ and applying Algorithm \ref{alg:debt} to do so, 
can under certain circumstances lead to a non-monotone series, as one can simply verify by a self-chosen example.
This makes it difficult to prove the convergence of such a series in general. 
Nevertheless, given a debt vector $\bar\bfr$ and $\bfs\ge\bzero_n$, we can still calculate the fixed point 
by avoiding Algorithm \ref{alg:debt} and use a Picard-type algorithm instead. 

\begin{algorithm}{5A}[Picard Iteration for the Debt Component]\label{alg:picard_debt}
Suppose that $\bfs\ge\bzero_n$ and $\eps > 0$.
\begin{enumerate}
  \item For $k=0$, set $\bfr^0=\bar\bfr$.  
  \item\label{alg:fp_picard_debt}
  For $k\ge 1$, determine $\bfr^k = \Phi^{\bfd}(\bfr^{k-1};\bfs)$.
  \item If $\|\bfr^{k-1}-\bfr^k\|<\eps$, stop the algorithm. 
  Else, set $k = k + 1$ and proceed with step \ref{alg:fp_picard_debt}.
\end{enumerate}
\end{algorithm}

\begin{proposition}\label{prop:debt_comp_conv_picard}
Algorithm \ref{alg:picard_debt} delivers a series of decreasing vectors $\bfr^k$ if $\Phi^{\bfd}(\bar\bfr;\bfs)\le \bar\bfr$ 
and a series of increasing vectors if $\Phi^{\bfd}(\bar\bfr;\bfs)\ge \bar\bfr$. 
Both series converge to the unique fixed point of $\Phi^{\bfd}(\cdot;\bfs)$. 
\end{proposition}

\begin{proof}
Assume that $\Phi^{\bfd}(\bar\bfr;\bfs)\ge \bar\bfr=\bfr^0$. 
For the first iterate, it holds that $\bfr^1=\Phi^{\bfd}(\bar\bfr;\bfs)\ge\bfr^0$. 
Via induction, it follows that $\bfr^{k+1}\ge\bfr^k$ for all $k\ge1$. 
Because the monotone series $\bfr^k$ is bounded by $\bfd$, it must converge to some fixed point. 
Because of the fact that the Elsinger condition holds, it follows directly that this fixed must be unique (see \citet{elsinger09}, Theorem 3). 
The argumentation is similar if $\Phi^{\bfd}(\bar\bfr;\bfs)\le \bar\bfr$.
\end{proof}

The Algorithms \ref{alg:debt} and \ref{alg:picard_debt} both enable us to calculate a new debt iterate given an equity vector. 
Together with Algorithm \ref{alg:equity} for the equity component, we can now combine both procedures in a common algorithm 
that searches for the fixed point $\bfR^*$.

\begin{algorithm}{6}[Hybrid Algorithm]\label{alg:comb_method}
Set $\eps > 0$. \hfill
\begin{enumerate}
  \item For $k=0$, choose $\bfr^0\in\{\rgreat,\rsmall\}$ and determine $\bfs(\bfr^0)$ with Algorithm \ref{alg:equity}. 
  \item\label{alg:comb_method_fp}
  For $k\ge1$:
  \begin{enumerate}
    \item[2.1] Determine $\bfr^k$ using Algorithm \ref{alg:debt} if $\bfr^0=\rgreat$ 
    or using Algorithm \ref{alg:picard_debt} if $\bfr^0=\rsmall$ in both cases with $\bfs=\bfs(\bfr^{k-1})$. 
    \item[2.2] Determine $\bfs^k = \bfs(\bfr^k)$ using Algorithm \ref{alg:equity}.
  \end{enumerate}
  \item If $\left\|\begin{pmatrix}\bfr^{k-1}\\\bfs^{k-1}\end{pmatrix} - \begin{pmatrix}\bfr^k\\\bfs^k\end{pmatrix}\right\|<\eps$, 
  stop the algorithm. 
  Else, set $k = k + 1$ and proceed with step \ref{alg:comb_method_fp}.
\end{enumerate}
\end{algorithm}

For given $\bfr = \bfr^k, k\ge 0$, the Hybrid Algorithm determines $\bfs^k=\bfs(\bfr^k)$ as the correct equity value 
that solves \eqref{eq:liq_eq_equity} and for given $\bfs = \bfs(\bfr^k), k\ge 0$, 
it determines the correct debt value $\bfr^{k+1}$ that solves \eqref{eq:liq_eq_debt}. 
As such, conditional on the values determined in the previous step, the algorithm calculates an exact solution 
of either \eqref{eq:liq_eq_equity} or \eqref{eq:liq_eq_debt} in the next iteration step. 

\begin{proposition}\label{prop:comb_meth_conv}
The Hybrid Algorithm delivers a series of decreasing vectors if $\bfr^0=\rgreat$ that converges to the fixed point $\bfR^*$. 
In case of $\bfr^0=\rsmall$ the series is increasing with the same limit.
\end{proposition}

\begin{proof}
First, suppose that $\bfr^0=\rgreat$. 
We will first show by induction that the series decreases. 
For the induction start note that
\begin{equation}
\bfr^1 = \Phi^{\bfd}(\bfr^1;\bfs(\bfr^0)) = \min\{\bfd, \bfa + \bfMd\bfr^1 + \bfMs\bfs(\bfr^0)\} \le \bfd = \bfr^0.
\end{equation}
As mentioned in the proof of Proposition \ref{prop:elsinger_better_picard}, the equity vectors $\bfs(\bfr)$ are increasing in $\bfr$ 
which yields to $\bfs(\bfr^1)\le \bfs(\bfr^0)$. 
For the induction step, assume that for $k>1$ it holds that $\bfr^{k-1}\ge\bfr^k$ and consequently $\bfs(\bfr^{k-1})\ge\bfs(\bfr^k)$. 
Since $\bfr^k=\Phi^{\bfd}(\bfr^k;\bfs(\bfr^{k-1}))$ and because of
\begin{equation}
\begin{split}
\Phi^{\bfd}(\bfr^k;\bfs(\bfr^k)) &= \min\{\bfd, \bfa + \bfMd\bfr^k + \bfMs\bfs(\bfr^k)\} \\
                                 &\le \min\{\bfd, \bfa + \bfMd\bfr^k + \bfMs\bfs(\bfr^{k-1})\} \\
                                 &= \Phi^{\bfd}(\bfr^k;\bfs(\bfr^{k-1}))\\
                                 &= \bfr^k 
\end{split}								 
\end{equation}
the assumption \eqref{eq:assu_debt_decr} is fulfilled. 
The next iterate $\bfr^{k+1}$ emerges from a decreasing sequence produced by applying Algorithm \ref{alg:debt} beginning with $\bar\bfr=\bfr^k$. 
Hence $\bfr^{k+1}\le\bfr^k$ and thus $\bfs(\bfr^{k+1})\le\bfs(\bfr^k)$. 
Next step is to show that the series converges to $\bfR^*$. 
We have that the two sequences $((\bfr^{k+1})^t, (\bfs(\bfr^k))^t)^t$ and $((\bfr^k)^t, (\bfs(\bfr^k))^t)^t$ are both decreasing in $(\R_0^+)^{2n}$ 
and therefore converge to the same limit $((\bfr^*)^t, (\bfs^*)^t)^t\in(\R_0^+)^{2n}$. 
Because of the continuity of $\Phi^{\bfd}$ and $\Phi^{\bfs}$ it must hold that $\Phi^{\bfd}(\bfr^*,\bfs^*)=\bfr^*$ 
and $\Phi^{\bfs}(\bfs^*,\bfr^*)=\bfs^*$. 
Thus, $((\bfr^*)^t, (\bfs^*)^t)^t$ solves \eqref{eq:liq_eq_equity} and \eqref{eq:liq_eq_debt}. 
The proof for $\bfr^0=\rsmall$ is similar. 
\end{proof}

In Proposition \ref{prop:elsinger_better_picard}, we have shown that when using the Elsinger Algorithm, 
the iterates will always be nearer to the solution $\bfR^*$ than the corresponding iterates of the Picard Algorithm. 
This lead to the conclusion that the iteration number is minimized for the Elsinger Algorithm. 
The next Proposition shows the same when comparing the Elsinger and the Hybrid Algorithm and it will become clear 
that the Hybrid Algorithm will need less iteration steps to reach $\bfR^*$ than the Elsinger Algorithm. 

\begin{proposition}\label{prop:hybrid_better_elsinger}
As in Proposition \ref{prop:elsinger_better_picard}, we denote the iterates of the two algorithms with subscripts, 
where E stands for the Elsinger and H for the Hybrid Algorithm. 
\begin{enumerate}[(i)]
  \item For any iterate $k\ge 1$ it holds that $\bfR^k_{\rm E}\ge \bfR^k_{\rm H}$ if $\bfR^0 = (\rgreat^t, \bfs(\rgreat)^t)^t$ 
  and $\bfR^k_{\rm E}\le \bfR^k_{\rm H}$ when $(\rsmall^t, \bfs(\rsmall)^t)^t$ is the starting vector of both algorithms. 
  \item Let $\bfR^k$, $k\ge0$, be an iterate either of the Elsinger Algorithm or of the Hybrid Algorithm 
  that started with $\bfR^0 = (\rgreat^t, \bfs(\rgreat)^t)^t$. 
  Then $\bfR^{k+1}_{{\rm E}}(\bfr^k)\ge \bfR^{k+1}_{{\rm H}}(\bfr^k)$ for the next iterates which were calculated 
  with either the Elsinger or the Hybrid Algorithm starting from $\bfR^k$. 
  If $\bfR^0 = ((\rsmall^t, \bfs(\rsmall)^t)^t$, it holds that $\bfR^{k+1}_{{\rm E}}(\bfr^k) \le \bfR^{k+1}_{{\rm H}}(\bfr^k)$.
\end{enumerate}
\end{proposition}

\begin{proof}
\begin{enumerate}[(i)]
  \item Let $\bfR^0 = (\rgreat^t, \bfs(\rgreat)^t)^t$. 
  From Proposition \ref{prop:comb_meth_conv} we know that $\bfr^1_{\rm H}\le\bfr^0_{\rm H}=\bfd$ which yields to
  \begin{equation}
  \begin{split}
  \bfr^1_{\rm E} = \min\{\bfd, \bfa + \bfMd\bfd + \bfMs\bfs(\bfd)\} \ge 
                   \min\{\bfd, \bfa + \bfMd\bfr^1_{\rm H} + \bfMs\bfs(\bfd)\} = \bfr^1_{\rm H}.  
  \end{split}
  \end{equation}
  Further, since $\bfs(\bfr)$ is increasing in $\bfr$ (cf. Proposition \ref{prop:elsinger_better_picard}),  
  $\bfs^1_{\rm E} = \bfs(\bfr^1_{\rm E}) \ge \bfs(\bfr^1_{\rm H}) = \bfs^1_{\rm H}$, which completes the induction start.
  For the induction step, assume that it holds for $k>1$ that $\bfR^{k-1}_{\rm E}\ge\bfR^{k-1}_{\rm H}$. 
  Because of Proposition \ref{prop:comb_meth_conv}, $\bfr^{k-1}_{\rm H}\ge \bfr^k_{\rm H}$ and thus
  \begin{equation}
  \begin{split}
  \bfr^k_{\rm E} &= \min\{\bfd, \bfa + \bfMd\bfr^{k-1}_{\rm E} + \bfMs\bfs(\bfr^{k-1}_{\rm E})\} \\
   &\ge \min\{\bfd, \bfa + \bfMd\bfr^{k-1}_{\rm H} + \bfMs\bfs(\bfr^{k-1}_{\rm H})\} \\
   &\ge \min\{\bfd, \bfa + \bfMd\bfr^k_{\rm H} + \bfMs\bfs(\bfr^{k-1}_{\rm H})\} \\
   &= \bfr^k_{\rm H},
  \end{split} 
  \end{equation}
  where we again used the fact that $\bfs(\bfr)$ is increasing in $\bfr$ 
  from which follows that $\bfs(\bfr^{k-1}_{\rm E})\ge\bfs(\bfr^{k-1}_{\rm H})$ 
  and also $\bfs(\bfr^k_{\rm E})\ge\bfs(\bfr^k_{\rm H})$. 
  The proof when $\bfR^0 = (\rsmall^t, \bfs(\rsmall)^t)^t$ is completely analogous. 
  
  \item Let $\bfR^0 = (\rgreat^t, \bfs(\rgreat)^t)^t$ and $\bfR^k=\bfR^k_{\rm E}$. 
  Note that because of $\bfr^k_{\rm E}\le\bfr^{k-1}_{\rm E}$ is holds that
  \begin{equation}
  \begin{split}
  \Phi^{\bfd}(\bfr^k_{\rm E}; \bfs(\bfr^k_{\rm E})) &= \min\{\bfd, \bfa + \bfMd\bfr^k_{\rm E} + \bfMs\bfs(\bfr^k_{\rm E})\} \\ 
   &\le \min\{\bfd, \bfa + \bfMd\bfr^{k-1}_{\rm E} + \bfMs\bfs(\bfr^{k-1}_{\rm E})\} \\
   &= \bfr^k_{\rm E}.
  \end{split} 
  \end{equation}
  Therefore, the assumption in \eqref{eq:assu_debt_decr} is fulfilled which ensures that $\bfr^{k+1}_{\rm H}\le\bfr^k_{\rm E}$. 
  For the next iterate it follows that
  \begin{equation}
  \begin{split}
  \bfr^{k+1}_{\rm E} &= \min\{\bfd, \bfa + \bfMd\bfr^k_{\rm E} + \bfMs\bfs(\bfr^k_{\rm E})\} \\ 
   &\ge \min\{\bfd, \bfa + \bfMd\bfr^{k+1}_{\rm H} + \bfMs\bfs(\bfr^k_{\rm E})\} \\ 
   &= \bfr^{k+1}_{\rm H},
  \end{split} 
  \end{equation}
  which in turn implies $\bfs^{k+1}_{\rm E}\ge\bfs^{k+1}_{\rm H}$. 
  On the other hand, starting with $\bfR^k=\bfR^k_{\rm H}$ yields because of $\bfr^{k+1}_{\rm H}\le \bfr^k_{\rm H}$ to
  \begin{equation}
  \begin{split}
  \bfr^{k+1}_{\rm E} &= \min\{\bfd, \bfa + \bfMd\bfr^k_{\rm H} + \bfMs\bfs(\bfr^k_{\rm H})\} \\ 
   &\ge \min\{\bfd, \bfa + \bfMd\bfr^{k+1}_{\rm H} + \bfMs\bfs(\bfr^k_{\rm H})\} \\
   &= \bfr^{k+1}_{\rm H}.
  \end{split} 
  \end{equation}
  If follows from this results that $\bfs^{k+1}_{\rm E}\le\bfs^{k+1}_{\rm H}$. 
  A similar argumentation together with Proposition \ref{prop:debt_comp_conv_picard} delivers the proof 
  in case of $\bfR^0 = (\rsmall^t, \bfs(\rsmall)^t)^t$.
\end{enumerate}
\end{proof}

Of course, within an iteration step of the Hybrid Algorithm, potentially many linear equation systems have to be solved 
since for the debt component Algorithm \ref{alg:debt} is applied, which results in higher computational costs. 
But if we ignore for a moment this circumstance it follows from Proposition \ref{prop:hybrid_better_elsinger} 
that the convergence speed of the Hybrid Algorithm is higher than the one of the Elsinger Algorithm.

\section{Finite Algorithms}\label{sec:def-set-algo}

The Algorithms in the previous section all had the drawback that it could not be ensured that the solution $\bfR^*$ 
is reached in a finite number of iteration steps. 
In this section we will present two ways in which potentially infinite solution algorithms can be turned into procedures 
that reach the solution in finitely many steps. 
The common principle of these methods is to include the information which firms are in default under a current iterate $\bfR^k$. 
It turns out that this slight modification helps to overcome the disadvantage of potentially infinitely many iteration steps.

To guarantee that the forthcoming procedures are well-defined, 
we have to drop the Elsinger Property and demand a stricter property of the ownership matrices (see \citet{fischer14}). 

\begin{assumption}\label{assu:os_mat_norm}
For both the debt and the equity ownership matrices it holds that $\|\bfMd\|<1$ and $\|\bfMs\|<1$.
\end{assumption}

For the remainder of this section we suppose that Assumption \ref{assu:os_mat_norm} holds. 
Note that Assumption \ref{assu:os_mat_norm} implies Assumption \ref{assu:holding_mat}, but not the other way round.
The financial system therefore still has a unique solution under Assumption \ref{assu:os_mat_norm}.

\begin{definition}\label{def:sys_sol}
Let $\bfR=(\bfr^t,\bfs^t)^t\in(\R_0^+)^{2n}$ be an arbitrary vector with corresponding default set $D(\bfR)$ and default matrix $\bfL=\bfL(\bfR)$. 
The \emph{pseudo solution $\widehat\bfR\in(\R_0^+)^{2n}$ of \eqref{eq:liq_eq_debt} 
and \eqref{eq:liq_eq_equity} that belongs to $D(\bfR)$} is defined by
\begin{equation}
\widehat\bfR = \begin{pmatrix} (\bfI_n-\bfL)\bfd + \bfL\bfx \\
(\bfI_n-\bfL)\bfx \end{pmatrix},
\end{equation}
where $\bfx\in\R^n$ is the solution of the linear equation system $\bfA\bfx=\bfb$ with 
\begin{equation}\label{eq:sys_sol_A}
\bfA = \bfI_n - \left(\bfMd\bfL + \bfMs(\bfI_n - \bfL)\right)\in\R^{n\times n} 
\end{equation}
and 
\begin{equation}\label{eq:sys_sol_b}
\bfb = \bfa + \bfMd(\bfI_n-\bfL)\bfd - (\bfI_n-\bfL)\bfd\in\R^n.
\end{equation}
\end{definition}

To motivate the definition of a pseudo solution, assume that it was known for each firm whether it was in default 
under the solution of \eqref{eq:liq_eq_debt} and \eqref{eq:liq_eq_equity} or not. 
Denote by $D^*\subseteq\mathcal N$ the set of firms that were in default under $\bfR^*$:
\begin{equation}
D^* = D(\bfr^*,\bfs^*) = \left\{i\in\mathcal N:a_i + \sum_{j=1}^n M^{\bfd}_{ij} r^*_j + \sum_{j=1}^n M^{\bfs}_{ij} s^*_j < d_i\right\}
\end{equation} 
and let $\bfL^*=\bfL(\bfr^*,\bfs^*)$ be the corresponding default matrix. 
We assume that the set was known even though this information is not available \emph{a priori}. 
However, if we had this information, no iteration procedure would be needed to find the fixed point $\bfR^*$. 
We only had to compute the pseudo solution that belongs to $D^*$, as is shown in Proposition \ref{prop:def_set_sol}.

The reason why we have to restrict the following considerations to ownership matrices with a matrix norm smaller one is 
because we have to guarantee that $\bfx$ from Definition \ref{def:sys_sol} is uniquely defined. 
This can only be ensured if $\|\bfMd\|<1$ and $\|\bfMs\|<1$ since then $\|\bfMd\bfL + \bfMs(\bfI_n - \bfL)\|<1$ for any $\bfL$ as well. 
This in turn implies that $\bfA$ in \eqref{eq:sys_sol_A} is invertible. 
For ownership matrices $\bfMd$ and $\bfMs$ that have the Elsinger Property, the invertibility of $\bfA$ is obviously not always given.

\begin{proposition}\label{prop:def_set_sol}
The pseudo solution belonging to $D^*$ is the solution $\bfR^*$ of the financial system $\mathcal F(\bfa, \bfMd, \bfMs, \bfd)$, i.e.
\begin{equation}
\bfR^* = \begin{pmatrix} (\bfI_n-\bfL^*)\bfd + \bfL^*\bfx \\
(\bfI_n-\bfL^*)\bfx \end{pmatrix},
\end{equation}
where $\bfL^*$ is the default matrix belonging to $D^*$ and $\bfx$ is the solution of the equation $\bfA\bfx=\bfb$ defined in 
\eqref{eq:sys_sol_A} and \eqref{eq:sys_sol_b}. 
\end{proposition}

\begin{proof}
According to the liquidation value equations in \eqref{eq:liq_eq_debt} and \eqref{eq:liq_eq_equity}, the vectors $\bfr^*$ and $\bfs^*$ are given as
\begin{equation}
r_i^* = \begin{cases} d_i ,& \text{if $i\notin D^*$,} \\
                      a_i + \sum_{j=1}^nM^{\bfd}_{ij}r^*_j + \sum_{j=1}^nM^{\bfs}_{ij}s^*_j ,&  \text{if $i\in D^*$} \end{cases}
\end{equation}
and 
\begin{equation}
s_i^* = \begin{cases} a_i + \sum_{j=1}^nM^{\bfd}_{ij}r^*_j + \sum_{j=1}^nM^{\bfs}_{ij}s^*_j - d_i ,&  \text{if $i\notin D^*$,} \\
                      0 ,&  \text{if $i\in D^*$}. \end{cases}
\end{equation}
In matrix notation this means in particular that $(\bfI_n-\bfL^*)\bfr^*=(\bfI_n-\bfL^*)\bfd$ and $\bfL^*\bfs^* = \bzero_n$ and thus
\begin{equation}
\bfR^* = \begin{pmatrix} (\bfI_n-\bfL^*)\bfd + \bfL^*\bfr^* \\
(\bfI_n-\bfL^*)\bfs^* \end{pmatrix}.
\end{equation}
For the firms in default we only have to calculate the debt payments 
and for the firms not in default we have to determine the equity value. 
The solution $\bfR^*$ does hence contain only $n$ unknown values 
and we only have to consider the two subsystems
\begin{equation}
\bfL^*\bfr^* = \bfL^*\bfa + \bfL^*\bfMd\bfr^*+\bfL^*\bfMs\bfs^*
\end{equation}
and
\begin{equation}
(\bfI_n-\bfL^*)\bfs^* = (\bfI_n-\bfL^*)(\bfa + \bfMd\bfr^*+\bfMs\bfs^* - \bfd)
\end{equation}
We can add the two equations and write the system more compact as:
\begin{equation}
\bfL^*\bfr^* + (\bfI_n-\bfL^*)\bfs^* = \bfa + \bfMd\bfr^*+ \bfMs\bfs^* - (\bfI_n-\bfL^*)\bfd.
\end{equation}
Because of $(\bfI_n-\bfL^*)\bfs^*=\bfs^*$ we get
\begin{equation}
\bfL^*\bfr^* + (\bfI_n-\bfL^*)\bfs^* = \bfa + \bfMd\bfr^*+ \bfMs(\bfI_n-\bfL^*)\bfs^* - (\bfI_n-\bfL^*)\bfd,
\end{equation}
which leads after some rearrangements to 
\begin{equation}
\bfL^*\bfr^* + (\bfI_n-\bfL^*)\bfs^* - \bfMd\bfL^*\bfr^* - \bfMs(\bfI_n-\bfL^*)\bfs^* = \bfa + \bfMd(\bfI_n-\bfL^*)\bfr^* - (\bfI_n-\bfL^*)\bfd
\end{equation}
that is equivalent to
\begin{equation}
\left(\bfI_n - \left(\bfMd\bfL^* + \bfMs(\bfI_n - \bfL^*)\right)\right) (\bfL^*\bfr^* + (\bfI_n-\bfL^*)\bfs^*) 
= \bfa + \bfMd(\bfI_n-\bfL^*)\bfd - (\bfI_n-\bfL^*)\bfd, 
\end{equation}
since $\bfL^*(\bfI_n-\bfL^*)=\bzero_{n\times n}$. 
Setting $\bfx=\bfL^*\bfr^* + (\bfI_n-\bfL^*)\bfs^*$ and with the notation of Definition \ref{def:sys_sol}, the equation system becomes $\bfA\bfx=\bfb$.
\end{proof}

The main challenge in this solution approach is of course that the final default set $D^*$ is unknown. 
Algorithms that follow this idea to find $\bfR^*$ consequently have to find $D^*$ in a fast way. 
A naive strategy could be to check all possible default scenarios of the financial system, 
calculate the pseudo solution for the corresponding default set and check whether it actually is the fixed point of $\Phi$. 
However, there are $2^n$ possible scenarios that would have to be checked, which could be cumbersome for large $n$. 
Therefore, more efficient algorithms are needed that require less computation to find $D^*$.
Some possible algorithms are presented in the next subsections .

\subsection{Decreasing Trial-and-Error Algorithms}\label{subsec:trial_error_alg_decr}

The three algorithms \ref{alg:picard}, \ref{alg:elsinger} and \ref{alg:comb_method} 
from Section \ref{sec:iterative_algo} can start with a vector $\bfR^0$ that was the upper boundary of the solution vector $\bfR^*$. 
The procedures in this subsection have in common that they also start with this upper boundary and calculate a corresponding default set. 
For every following iterate, the corresponding default set is determined as well. 
To avoid that every default set it is checked whether it actually is $D^*$ 
and whether the corresponding pseudo solution is the fixed point of $\Phi$, 
the algorithm will identify potential default sets to reduce the computational effort. 
If it turns out that the potential default set is $D^*$, the algorithm stops. 
Otherwise, the procedure continues until a new potential default set is found that has to be checked again, and so on.
Due to these characteristics we name this type of algorithm \emph{Trial-and-Error Algorithm}. 
The general procedure of algorithms of this type is similar. 

\begin{algorithm}{7}[Decreasing Trial-and-Error Algorithm]\label{alg:trial_error_decr}
Set $l\ge2$ and $p=0$.\hfill
\begin{enumerate}
  \item\label{alg:t-e_decr_choice}
  Choose either the Picard (Algorithm \ref{alg:picard}), the Elsinger (Algorithm \ref{alg:elsinger}) 
  or the Hybrid Algorithm (Algorithm \ref{alg:comb_method}) which is used in the following to generate the next iterate. 
  \item If in Step \ref{alg:t-e_decr_choice} the Picard Algorithm is chosen, set $d=-1$, $\bfR^0=\Rgreat$ and determine $D(\bfR^0)$. 
  Else, set $d=0$, $\bfR^0=\left(\begin{smallmatrix}\bfd \\ \bfs(\bfd)\end{smallmatrix}\right)$ and determine $D(\bfR^0)$.
  \item\label{alg:t-e_decr_all_def}
  If $D(\bfR^0)=\mathcal N$, set $\bfR^* = \left(\begin{smallmatrix}(\bfI_n-\bfMd)^{-1}\bfa\\\bzero_n\end{smallmatrix}\right)$ 
  and stop the algorithm.
  \item\label{alg:t-e_decr_no_def}
  If the Elsinger or the Hybrid Algorithm is chosen in Step \ref{alg:t-e_decr_choice} and if $D(\bfR^0)=\emptyset$, 
  set $\bfR^*=\bfR^0$ and stop the algorithm. 
  \item\label{alg:trial_error_fp_decr}
  Else, calculate for $k > p$ the iterates $\bfR^k$ starting with $\bfR^p$ using the algorithm chosen in step \ref{alg:t-e_decr_choice}
  and the corresponding default sets $D(\bfR^k)$ until $k=q$ with
  \begin{equation}\label{eq:pot_def_set}
  q = \min\{m>p: D(\bfR^{m-l+1}) = \ldots = D(\bfR^m)\ \text{and}\ |D(\bfR^m)| > d\}
  \end{equation}
  is reached. 
  Determine the pseudo solution belonging to $D(\bfR^q)$ and denote it by $\widehat\bfR^q$.
  \item If $\Phi(\widehat\bfR^q) = \widehat\bfR^q$, stop the algorithm. 
  Else, set $d = |D(\bfR^q)|$ and $p=q$ and proceed with step \ref{alg:trial_error_fp_decr}.  
\end{enumerate}
\end{algorithm}

The Algorithms \ref{alg:picard}, \ref{alg:elsinger} and \ref{alg:comb_method} in their decreasing versions 
produce decreasing sequences of iterates and thus increasing sequences of default sets, i.e. $D(\bfR^k)\subseteq D(\bfR^{k+1})$ for $k\ge0$.
Algorithm \ref{alg:trial_error_decr} means that one iterates 
and checks whether the default set has not changed compared to the previous default set. 
If the default set stays the same for the next $l-1$ consecutive iterations, 
this is an indication that the actual $D^*$ might have been reached. 
To check this, the pseudo solution is calculated and it is checked whether it solves \eqref{eq:liq_eq_debt} and \eqref{eq:liq_eq_equity}. 
If no solution has been found, one iterates again until a larger default sets stays identical for $l-1$ consecutive times, 
and the described procedure can be repeated. 
If a solution is reached, the procedure stops. 
Due to its described property, we call $l$ the \emph{lag value}.

In the special case of $l=2$ this means that the pseudo solution is calculated if the default set stays the same from one iteration step to another. 
Obviously, choosing a higher lag value inspires more confidence in the potential default set since the longer the default set stays unchanged, 
the higher is the chance that it is the actual default set. 

Depending on the choice of the algorithm in Step \ref{alg:t-e_decr_choice}
of the Decreasing Trial-and-Error Algorithm, we obtain three different versions of Algorithm \ref{alg:trial_error_decr}:
\begin{itemize}
  \item[(i)] The \emph{Decreasing Trial-and-Error Picard Algorithm} with $\bfR^0 = \Rgreat$, where the iterates as given by $\bfR^k=\Phi(\bfR^{k-1})$.
  \item[(ii)] The \emph{Decreasing Trial-and-Error Elsinger Algorithm} with $\bfR^0 = ((\rgreat)^t, (\bfs(\rgreat))^t)^t$, 
  where $\bfs(\rgreat)$ is obtained via Algorithm \ref{alg:equity} and the next iterates are obtained using Algorithm \ref{alg:elsinger}. 
  \item[(iii)] The \emph{Decreasing Trial-and-Error Hybrid Algorithm} with the same starting vector as in (ii),  
  where the next iterates are obtained using Algorithm \ref{alg:comb_method}.
\end{itemize}

The particular cases when $D(\bfR^0)\in\{\emptyset,\mathcal N\}$ in the steps \ref{alg:t-e_decr_all_def} and \ref{alg:t-e_decr_no_def}, 
deserve a separate mention since in such situations, no iteration is necessary and the solution $\bfR^*$ can be given explicitly 
under some circumstances.
The justification of this phenomena is given in the following proposition. 

\begin{proposition}\label{prop:hybrid_zero_all}
For the Decreasing Trial-and-Error Hybrid Algorithm the following holds:
\begin{enumerate}
  \item[(i)] If $D(\bfR^0)=\mathcal N$, then $\bfR^*=\left(\begin{smallmatrix}(\bfI_n-\bfMd)^{-1}\bfa\\\bzero_n\end{smallmatrix}\right)$, 
  no matter which version of the algorithm is taken. 
  \item[(ii)] If $D(\bfR^0)=\emptyset$ and either the Decreasing Trial-and-Error Elsinger Algorithm or 
  the Decreasing Trial-and-Error Hybrid Algorithm is used, then $\bfR^0=\bfR^*$.
\end{enumerate}
\end{proposition}

\begin{proof}
\begin{enumerate}
  \item[(i)] 
  First, assume that the Picard Algorithm is chosen in Step \ref{alg:t-e_decr_choice} of Algorithm \ref{alg:trial_error_decr}. 
  Because of $D(\bfR^0)=\mathcal N$, it must hold that $\bfa + \bfMd\bfd + \bfMs\sgreat < \bfd$ and also $\bfa + \bfMd\bfd < \bfd$. 
  A consequence is that $\sgreat = (\bfI_n-\bfMs)^{-1}(\bfa + \bfMd\bfd - \bfd)^+=\bzero_n$. 
  From Proposition \ref{prop:lim_R} it follows that $\bfs^*=\bzero$. 
  For $\bfs=\bfs^*=\bzero$, Equation \eqref{eq:liq_eq_debt} is now solved by $\bfr^*=(\bfI_n - \bfMd)^{-1}\bfa$, 
  where Lemma \ref{lem:inv_xos-mat} proves that $(\bfI_n - \bfMd)^{-1}$ exists. 
  If the Elsinger or the Hybrid Algorithm is chosen in Step \ref{alg:t-e_decr_choice}, we have that $\bfa + \bfMd\bfd + \bfMs\bfs(\bfd) < \bfd$. 
  It follows from \eqref{eq:phi_aux_eq_fp} that $\bfs(\bfd)=\bfs^*=\bzero_n$  since $\bfs^*=\bfs(\bfr^*)\le\bfs(\bfd)$ 
  because of $\bfr^*\le\bfd$ and the fact that $\bfs(\bfr)$ is increasing in $\bfr$. 
  The solution of Equation \eqref{eq:liq_eq_debt} is therefore the same as in the Picard case. 
  \item[(ii)] 
  Now, $\bfR^0=\left(\begin{smallmatrix}\bfd \\ \bfs(\bfd)\end{smallmatrix}\right)$ and since $D(\bfR^0)=\emptyset$, 
  it holds that $\bfa + \bfMd\bfd + \bfMs\bfs(\bfd)\ge\bfd$. 
  This leads to 
  \begin{equation}\label{eq:fp_t-e_decr_no_def}
  \Phi\begin{pmatrix} \bfd \\ \bfs(\bfd) \end{pmatrix} = 
  \Phi\begin{pmatrix} \min\{\bfd, \bfa + \bfMd\bfd + \bfMs\bfs(\bfd)\} \\ (\bfa + \bfMd\bfd + \bfMs\bfs(\bfd) - \bfd)^+ \end{pmatrix} = 
  \begin{pmatrix} \bfd \\ \bfs(\bfd) \end{pmatrix},
  \end{equation}
  which proves the claim. 
\end{enumerate}
\end{proof} 

Note that for the Decreasing Trial-and-Error Picard Algorithm, we cannot conclude that $\bfR^0=\Rgreat=\bfR^*$ if $D(\bfR^0)=\emptyset$. 
There are simple counterexamples for situations like this. 

\begin{proposition}\label{prop:conv_def_set}
Algorithm \ref{alg:trial_error_decr} reaches the solution $\bfR^*$ of \eqref{eq:liq_eq_debt} 
and \eqref{eq:liq_eq_equity} in a finite number of iteration steps.
\end{proposition}

\begin{proof}
By definition of $D(\bfR^k)$ in \eqref{eq:defi_def_set} and since $\bfR^k$ converges to $\bfR^*$ from above for any of the three algorithms 
\ref{alg:picard}, \ref{alg:elsinger} and \ref{alg:comb_method}, 
there exists a $k^0\ge0$ such that $D(\bfR^k)=D(\bfR^*)=D^*$ for all $k\ge k^0$.
\end{proof}

\subsection{Increasing Trial-and-Error Algorithms}\label{subsec:trial_error_alg_incr}

In contrast to the decreasing algorithms presented in the subsection above, it is of course also possible to use an algorithm 
with the reverse direction, i.e. in which the series of produced iterates is increasing and in which the default sets are decreasing. 
The general form is very similar to Algorithm \ref{alg:trial_error_decr}.

\begin{algorithm}{8}[Increasing Trial-and-Error Algorithm]\label{alg:trial_error_incr}
Set $l\ge2$, $d=n+1$ and $p=0$.\hfill
\begin{enumerate}
  \item Choose a starting vector $\bfR^0$ and determine $D(\bfR^0)$. 
  \item\label{alg:t-e_incr_no_def}
  If $D(\bfR^0)=\emptyset$, set $\bfR^*=\left(\begin{smallmatrix}\bfd \\ \bfs(\bfd)\end{smallmatrix}\right)$ and stop the algorithm.
  \item\label{alg:trial_error_fp_incr}
  Else, calculate for $k > p$ the iterates $\bfR^k$ starting with $\bfR^p$ using one of the Algorithms \ref{alg:picard}, \ref{alg:elsinger} 
  or \ref{alg:comb_method} and the corresponding default sets $D(\bfR^k)$ until $k=q$ with
  \begin{equation}\label{eq:pot_def_set_incr}
  q = \min\{m>p: D(\bfR^{m-l+1}) = \ldots = D(\bfR^m)\ \text{and}\ |D(\bfR^m)| < d\}
  \end{equation}
  is reached. 
  Determine the pseudo solution belonging to $D(\bfR^q)$ and denote it by $\widehat\bfR^q$.
  \item If $\Phi(\widehat\bfR^q) = \widehat\bfR^q$, stop the algorithm. 
  Else, set $d = |D(\bfR^q)|$ and $p=q$ and proceed with step \ref{alg:trial_error_fp_incr}.  
\end{enumerate}
\end{algorithm}

The functioning of Algorithm \ref{alg:trial_error_incr} is similar to the Decreasing Trial-and-Error Algorithms with the difference 
that the resulting sequence of default sets is obviously decreasing. 
As in Section \ref{subsec:trial_error_alg_decr}, the way of choosing the calculation method to determine the next iterate, 
allows three different modifications:
\begin{itemize}
  \item[(i)] The \emph{Increasing Trial-and-Error Picard Algorithm} with $\bfR^0=\Rsmall$ and $\bfR^k=\Phi(\bfR^{k-1})$.
  \item[(ii)] The \emph{Increasing Trial-and-Error Elsinger Algorithm} with $\bfR^0=((\rsmall)^t, (\bfs(\rsmall))^t)^t$, 
  where $\bfs(\rsmall)$ is obtained via Algorithm \ref{alg:equity} and the next iterates are obtained using Algorithm \ref{alg:elsinger}.
  \item[(iii)] The \emph{Increasing Trial-and-Error Hybrid Algorithm} with the same starting vector as in (ii) 
  and where the next iterates are obtained using Algorithm \ref{alg:comb_method}.
  Note that for the next debt iterate, Algorithm \ref{alg:picard_debt} is used instead of Algorithm \ref{alg:debt} in the decreasing version.
\end{itemize}

The justification of the stopping criteria in Step \ref{alg:t-e_incr_no_def} of Algorithm \ref{alg:trial_error_incr} is as follows. 
Suppose that the Picard version of the algorithm is chosen and that $D(\bfR^0)=\emptyset$, 
which means that $\bfa+\bfMd\rsmall+\bfMs\ssmall \ge \bfd$. 
Since $\rsmall\le\bfd$ and $\ssmall\le\bfs(\bfd)$, it also holds that $\bfa+\bfMd\bfd+\bfMs\bfs(\bfd) \ge \bfd$ 
and $\bfs(\bfd)\ge\bzero_n$ following from this. 
With Equation \eqref{eq:fp_t-e_decr_no_def}, we see that $\bfR^*=\left(\begin{smallmatrix}\bfd \\ \bfs(\bfd)\end{smallmatrix}\right)$. 
Also note that, in contrast to Algorithm \ref{alg:trial_error_decr}, there is no stopping criteria in case of $D(\bfR^0)=\mathcal N$. 
The reason is that in this case, no general statement about the structure of the solution $\bfR^*$ can be made, 
no matter which version of the algorithm is used.
In particular, from $D(\bfR^0)=\mathcal N$ it does not follow in general that $D(\bfR^*)=\mathcal N$, 
since there are easily constructable counterexamples for this. 

The reason why we distinguish between decreasing and increasing Trial-and-Error Algorithms 
is that Algorithm \ref{alg:trial_error_decr} will always find the correct default set $D^*=D(\bfR^*)$, 
and this in a finite number of iteration steps. 
For the Increasing Trial-and-Error Algorithms such a statement is not possible in general 
since there are some situations in which the default sets do not converge to $D^*$, no matter which lag value is chosen. 
Situations in which this ``anomaly'' occurs are always financial systems that contain a so-called \emph{borderline firm}. 
The expression borderline is taken from \citet{liu10} and denotes a firm $i\in\mathcal N$ in a financial system with fixed point $\bfR^*$ 
for which it holds that $r^*_i=d_i$ and $s^*_i=0$. 
In other words, borderline firms are just able to fully cover their liabilities, but have no remaining capital left in their balance sheet 
that can be furnished to their shareholders. 
By definition of a default set in \eqref{eq:defi_def_set}, a borderline firm $i$ is not in default since
\begin{equation}
0 = s_i = a_i + \sum_{j=1}^nM_{ij}^{\bfd}r^*_j + \sum_{j=1}^nM_{ij}^{\bfs}s^*_j - d_i 
\end{equation}
and therefore $i\notin D(\bfR^*)$.
However, when using an Increasing Trial-and-Error Algorithm it can happen for such a borderline firm $i$ 
that $i\in D(\bfR^k)$ for every iterate $\bfR^k$, $k\ge0$. 
This means that the true default set $D^*$ will never be identified by the algorithm.  
There exist many examples of financial systems that have this property.
To show that in such situations, the fixed point $\bfR^*$ can still be determined via the calculation of the pseudo solution, 
assume that the set $\mathcal B\subset \mathcal N$ contains an arbitrary selection of borderline firms. 
The common set of defaulting firms and the selected borderline firms is denoted by $\widetilde D$, i.e. $\widetilde D=D^* \cup \mathcal B$. 
The corresponding ``default'' matrices are given by $\widetilde\bfL=\widetilde\bfL(\widetilde D)$ and $\bfL^*=\bfL^*(D^*)$, respectively. 
Following this notation, $\widetilde\bfA$ and $\bfA^*$ define the matrices from \eqref{eq:sys_sol_A} with the corresponding default matrix, 
and $\widetilde\bfb$ and $\bfb^*$ are defined analogously. 
Moreover, we define $\widetilde\bfL_{\mathcal B}=\widetilde\bfL-\bfL^*$ as the diagonal matrix that indicates only the selected borderline firms.

\begin{lemma}\label{lem:pseudo_equiv}
The vector $\bfx^*=\bfL^*\bfr^* + (\bfI_n-\bfL^*)\bfs^*$ solves the equation system $\bfA^*\bfx=\bfb^*$ if and only if 
$\widetilde\bfx = \bfx^* + \widetilde\bfL_{\mathcal B}\bfd$ is the solution of $\widetilde\bfA\bfx=\widetilde\bfb$.
\end{lemma} 

\begin{proof}
Without loss of generality, we assume that the first $n_1$ firms of the system are solvent, 
that the next $n_2-n_1$ firms are the selected borderline cases, and that the remaining firms are in default under $\bfR^*$. 
This means that
\begin{equation}
\mathcal N = \{1\ld n_1\} \cup \mathcal B \cup D^* = \{1\ld n_1\} \cup \{n_1+1\ld n_2\} \cup \{n_2+1\ld n\}.
\end{equation}
It follows from Proposition \ref{prop:def_set_sol} that 
\begin{equation}
\bfx^* = (x_1\ld x_{n_1}, x_{n_1+1}\ld x_{n_2}, x_{n_2+1}\ld x_n)^t = (s^*_1\ld s^*_{n_1}, 0\ld 0, r^*_{n_2+1}\ld r^*_n)^t.
\end{equation}
Further, note that 
\begin{equation}
\bfMs(\bfI_n-\bfL^*)\bfx^*=\bfMs(\bfI_n-\widetilde\bfL)\bfx^* = \bfMs(\bfI_n-\widetilde\bfL)(\bfx^* + \widetilde\bfL_{\mathcal B}\bfd)
\end{equation}
and 
\begin{equation}
\bfMd(\bfI_n-\bfL^*)\bfd - \bfMd\widetilde\bfL_{\mathcal B}\bfd = \bfMd(\bfI_n-\widetilde\bfL)\bfd
\end{equation}
and that 
\begin{equation}
\bfMd\bfL^*\bfx^*+\bfMd\widetilde\bfL_{\mathcal B}\bfd=\bfMd\widetilde\bfL(\bfx^* + \widetilde\bfL_{\mathcal B}\bfd)
\end{equation}
because of the structure of $\bfx^*$. 
By \eqref{eq:sys_sol_A} and \eqref{eq:sys_sol_b}, $\bfx^*$ solves $\bfA^*\bfx=\bfb^*$ if and only if
\begin{equation}
\begin{split}
\bfx^* &= \bfb^* + (\bfMd\bfL^* + \bfMs(\bfI_n-\bfL^*))\bfx^* \\
       &= \bfa + \bfMd(\bfI_n-\bfL^*)\bfd - (\bfI_n-\bfL^*)\bfd + (\bfMd\bfL^* + \bfMs(\bfI_n-\bfL^*))\bfx^* \\
       &= \bfa + \bfMd(\bfI_n-\bfL^*)\bfd - \bfMd\widetilde\bfL_{\mathcal B}\bfd - (\bfI_n-\bfL^*)\bfd + 
       (\bfMd\bfL^* + \bfMs(\bfI_n-\bfL^*))\bfx^* + \bfMd\widetilde\bfL_{\mathcal B}\bfd \\
       &= \bfa + \bfMd(\bfI_n-\widetilde\bfL)\bfd - (\bfI_n-\bfL^*)\bfd + (\bfMd\widetilde\bfL + 
       \bfMs(\bfI_n-\widetilde\bfL))(\bfx^* + \widetilde\bfL_{\mathcal B}\bfd).
\end{split}
\end{equation}
Since $(\bfI_n-\bfL^*)\bfd =(\bfI_n-\widetilde\bfL)\bfd + \widetilde\bfL_{\mathcal B}\bfd$, 
we can add $\widetilde\bfL_{\mathcal B}\bfd$ on both sides of the equation and obtain
\begin{equation}
\bfx^* + \widetilde\bfL_{\mathcal B}\bfd = \bfa + \bfMd(\bfI_n-\widetilde\bfL)\bfd - (\bfI_n-\widetilde\bfL)\bfd + (\bfMd\widetilde\bfL + 
\bfMs(\bfI_n-\widetilde\bfL))(\bfx^* + \widetilde\bfL_{\mathcal B}\bfd),
\end{equation}
Therefore $\widetilde\bfx=(s^*_1\ld s^*_{n_1}, d_{n_1+1}\ld d_{n_2}, r^*_{n_2+1}\ld r^*_n)^t$ 
is the solution of $\widetilde\bfA\bfx=\widetilde\bfb$ if and only if $\bfx^*$ solves $\bfA^*\bfx=\bfb^*$. 
\end{proof}

The pseudo solution belonging to $D^*$ is the solution $\bfR^*$ of the system. 
A direct consequence of Lemma \ref{lem:pseudo_equiv} is that the pseudo solution of $\widetilde D$ is also equal to $\bfR^*$. 
Similar to the proof of Proposition \ref{prop:conv_def_set}, we can argue that the set $\widetilde D$ will be reached 
by the Increasing Trial-and-Error Algorithms in a finite number of steps. 
Note that this statement holds in particular for the Increasing Hybrid Trial-and-Error Algorithm, 
where Algorithm is used \ref{alg:picard_debt} to calculate the next debt iterate. 
Even though a Picard-typed procedure is used in this auxiliary algorithm, we can conclude together with 
Proposition \ref{prop:hybrid_better_elsinger} and $\eps>0$ that the number of iteration will still be finite. 
We summarize the findings in the next proposition. 

\begin{proposition}\label{prop:conv_def_set_incr}
Algorithm \ref{alg:trial_error_incr} reaches the solution $\bfR^*$ of \eqref{eq:liq_eq_debt} 
and \eqref{eq:liq_eq_equity} in a finite number of iteration steps.
\end{proposition}

\subsection{Sandwich Algorithms}

A disadvantage of the Trial-and-Error Algorithms was that when a potential default set is reached, 
the only way to find out whether this default set is actually $D^*$, is to calculate the corresponding pseudo solution 
and check whether it is a fixed point of \eqref{eq:Phi}.
The choice of a high lag value can of course increase the chance that $D^*$ is reached at the first trial, 
but there is no certainty. 

Another way to find $D^*$ is to start an iteration simultaneously with the largest and smallest possible solution 
and use one of the Algorithms \ref{alg:picard}, \ref{alg:elsinger} or \ref{alg:comb_method} to obtain the next iterate. 
For $k\ge0$ denote by $\overline{\bfR}^k$ the $k$-th iterate of the series that emerges when starting the algorithm with the maximum 
and by $\underline{\bfR}^k$ its counterpart when starting with the minimum possible solution. 
Depending which algorithm is chosen, the starting vector can either be $\overline{\bfR}^0=\Rgreat$ (Picard Iteration) or 
$\overline{\bfR}^0=(\rgreat^t, (\bfs(\rgreat))^t)^t$ (Elsinger and Hybrid Algorithm) when starting with the upper boundary. 
Analogously, we have $\underline{\bfR}^0=\Rsmall$ or $\underline{\bfR}^0=(\rsmall^t, (\bfs(\rsmall))^t)^t$ 
if the minimum possible solution is the starting point. 
By the Propositions \ref{prop:conv_picard}, \ref{prop:elsinger_conv}, \ref{prop:comb_meth_conv} and by Equation \eqref{eq:defi_def_set}, 
the iterative use of one of the mentioned algorithms entails that the default sets approach one another, i.e. for $k\ge0$
\begin{equation}
D(\underline{\bfR}^k) \supseteq D(\underline{\bfR}^{k+1}) \supseteq D^* \supseteq D(\overline{\bfR}^{k+1}) \supseteq D(\overline{\bfR}^k).
\end{equation}
Let 
\begin{equation}\label{eq:lag_sandwich}
l = \min\{k\ge0:D(\underline{\bfR}^k) = D(\overline{\bfR}^k)\}
\end{equation}
be the first iteration step in which the default set for both starting vectors is the same. 
Then we must have that $D(\overline{\bfR}^l)=D^*$ and, by Proposition \ref{prop:def_set_sol}, 
determining the pseudo solution belonging to $D^*$ leads to $\bfR^*$. 
Because of its characteristics we call this algorithm the \emph{Sandwich Algorithm}. 

\begin{algorithm}{9}[Sandwich Algorithm]\label{alg:sandwich}\hfill
\begin{enumerate}
  \item Determine $\overline{\bfR}^0$ and $\underline{\bfR}^0$ as well 
  as their corresponding default sets $D(\overline{\bfR}^0)$ and $D(\underline{\bfR}^0)$.
  \item\label{alg:fp_sandwich}
  For $k\ge 1$, calculate the iterates $\overline{\bfR}^k$ and $\underline{\bfR}^k$ 
  using one of the Algorithms \ref{alg:picard}, \ref{alg:elsinger} or \ref{alg:comb_method} 
  and the corresponding default sets $D(\overline{\bfR}^k)$ and $D(\underline{\bfR}^k)$. 
  \item If $D(\overline{\bfR}^k)=D(\underline{\bfR}^k)$, stop the algorithm, 
  set $D^*=D(\overline{\bfR}^k)$ and calculate the pseudo solution that belongs to $D^*$ following Definition \ref{def:sys_sol}.
  Else, set $k=k+1$ and go back to step \ref{alg:fp_sandwich}.
\end{enumerate}
\end{algorithm}

As for the Trial-and-Error Algorithms in the sections above, the Sandwich Algorithm results in different versions:
\begin{enumerate}
  \item[(i)] The \emph{Sandwich Picard Algorithm} with $\overline{\bfR}^0=\Rgreat$ and $\underline{\bfR}^0=\Rsmall$ 
  and the use of Algorithm \ref{alg:picard} in step \ref{alg:fp_sandwich}. 
  \item[(ii)] The \emph{Sandwich El\-sing\-er Algorithm} with $\overline{\bfR}^0=(\rgreat^t,(\bfs(\rgreat))^t)^t$ 
  and $\underline{\bfR}^0=(\rsmall^t,(\bfs(\rsmall))^t)^t$ and the use of Algorithm \ref{alg:elsinger} in step \ref{alg:fp_sandwich}.
  \item[(iii)] The \emph{Sandwich Hybrid Algorithm} with the same starting points as the Sandwich El\-sing\-er Algorithm 
  and the iterative use of Algorithm \ref{alg:comb_method} in step \ref{alg:fp_sandwich}.
\end{enumerate}

Recall the insights from Section \ref{subsec:trial_error_alg_incr}, where it was shown that, under some circumstances, it may happen 
that the series of default sets $D(\underline{\bfR}^k)$ will never converge to the actual default set $D^*$. 
Situations in which this problem occurs always contain at least one firm  that is on borderline in the solution $\bfR^*$. 
As a result of this behavior, the Sandwich Algorithm may not converge in the sense 
that the default sets $D(\overline{\bfR}^k)$ and $D(\underline{\bfR}^k)$ will never become identical. 
However, if we consider a stochastic setting and assume a distribution for the vector $\bfa$ of the exogenous assets' prices 
which has a density with respect to the Lebesgue measure on $(\R_0^+)^n$,
then situations in which the convergence cannot be assured occur only with probability zero as the next Proposition shows. 
Note that this assumption is fulfilled in the usual $n$-firm Merton models where the individual $a_i$ are log-normally distributed.

\begin{proposition}\label{prop:sandwich}
The Sandwich Algorithm generates a sequence of decreasing default sets $D(\underline{\bfR}^k)$ 
and a sequence of increasing default sets $D(\overline{\bfR}^k)$ that reach the default set $D^*$ of the solution $\bfR^*$ 
almost surely after finitely many steps. 
Thus, it reaches the solution $\bfR^*$ of \eqref{eq:Phi} almost surely after finitely many steps.
\end{proposition}

\begin{proof}
The increasing and decreasing property of the default sets follows directly from the Propositions 
\ref{prop:conv_picard}, \ref{prop:elsinger_conv} and \ref{prop:comb_meth_conv}. 
The two series of default sets of the algorithm both converge in finitely many iteration steps to $D^*$ 
if there is no firm in the financial system that is borderline. 
Lemma \ref{lem:boderline_as} in the Appendix shows that the probability for borderline firms in $\bfR^*$ is zero 
from which almost sure convergence follows. 
\end{proof}

By its nature, the Sandwich Algorithm converges to $D^*$ from both directions which doubles the computation 
and makes the algorithm somewhat inefficient from a computational point of view. 
On the other hand, the algorithm computes an exact solution in finitely many iteration steps without wasting time on ``Trial-and-Error''.
In contrast to the Trial-and-Error Algorithms, the drawback of the Sandwich Algorithm is 
that the convergence of the procedure cannot be ensured when borderline firms are present in the system. 
To overcome this problem, we recommend for practical purposes to apply the idea of a lag value in the Sandwich Algorithm as well. 

\begin{algorithm}{10}[Modified Sandwich Algorithm]\label{alg:sandwich_mod} Set $l\ge2$.
\begin{enumerate}
  \item Determine $\overline{\bfR}^0$ and $\underline{\bfR}^0$ as well 
  as their corresponding default sets $D(\overline{\bfR}^0)$ and $D(\underline{\bfR}^0)$.
  \item\label{alg:fp_sandwich_mod}
  For $k\ge 1$, calculate the iterates $\overline{\bfR}^k$ and $\underline{\bfR}^k$ 
  using one of the Algorithms \ref{alg:picard}, \ref{alg:elsinger} or \ref{alg:comb_method} 
  and the corresponding default sets $D(\overline{\bfR}^k)$ and $D(\underline{\bfR}^k)$. 
  \item If $D(\overline{\bfR}^k)=D(\underline{\bfR}^k)$, stop the algorithm, 
  set $D^*=D(\overline{\bfR}^k)$ and calculate the pseudo solution that belongs to $D^*$ following Definition \ref{def:sys_sol}.
  Else, if $k\ge l$ and
  \begin{equation}\label{eq:lag_equal_sandwich}
  |D(\underline{\bfR}^k)| - |D(\overline{\bfR}^k)| = \ldots = |D(\underline{\bfR}^{k-l+1})| - |D(\overline{\bfR}^{k-l+1})|,
  \end{equation}
  calculate the pseudo solution belonging to $D(\overline{\bfR}^k)$ and stop the algorithm 
  if it solves the Equations \eqref{eq:liq_eq_debt} and \eqref{eq:liq_eq_equity}. 
  Else, set $k=k+1$ and go back to step \ref{alg:fp_sandwich_mod}.
\end{enumerate}
\end{algorithm}

The modification consists of interrupting the algorithm if the default sets $D(\overline{\bfR}^k)$ 
and $D(\underline{\bfR}^k)$ for both iteration directions are not identical but stay constant for $l$ consecutive times. 
If $l$ is chosen large enough (e.g. $l\ge5$) and \eqref{eq:lag_equal_sandwich} holds, 
this is a strong indication that at least one firm in the system is borderline and that the convergence of both series is not given. 
In this situation, a check whether the default set has already been reached is suitable.

\section{Simulation Study}\label{sec:simulation}

Aim of this section is to confirm the theoretical findings in the former sections by simulation. 
In particular, we focus our considerations on the following subjects:
\begin{enumerate}[(i)]
  \item Investigate the trade-off when choosing a lag value $l$ in the Trial-and-Error Algorithms of Section \ref{subsec:trial_error_alg_decr}. 
  The result of this part contains `optimal' lag values for each algorithm and will be used in the next part. 
  \item Investigate the algorithm efficiency of all the presented algorithms in the Sections \ref{sec:iterative_algo} and \ref{sec:def-set-algo}. 
  For every different technique to calculate the next iterate (Picard, Elsinger, Hybrid), we compare the three types of algorithms 
  (non-finite, Trial-and-Error, Sandwich) with each other. 
  \end{enumerate}
Before presenting the results of our study, the used financial systems are further specified.

\subsection{General Structure of the Financial Systems}

For the system size $n$ we chose six different values, viz. $n\in\{5, 10, 25, 50, 100, 200\}$. 
A system with only 5 or 10 firms can be considered as relatively small whereas networks with $n=25$ or $n=50$ are regarded as medium-sized. 
Small systems are investigated for example in \citet{gourieroux12}, \citet{rogers13} and \citet{elsinger06a} 
where the size was 5, 6 and 10 firms respectively.  
Examples of medium-sized systems are \citet{acemoglu13} and \citet{nier07} where the size was 20 and 25 firms respectively. 
Further, we added networks with 100 and 200 firms into our study to also include larger systems. 
Existing studies for such sizes are \citet{elliott13}, \citet{cont11} and \citet{gai11} 
that entailed networks with 100, 125 and 250 firms, respectively.

There are some empirical studies (\citet{elsinger06b}, \citet{gai10}) that investigated system sizes of about $n=1000$. 
We believe that for practical purposes such large systems are not of interest which is why we did not take values $n>200$ into account. 
However, it is expectable that the results obtained for our system sizes also hold for networks with more than 200 firms. 

The next input parameters to define are the asset and debt values. 
To keep it simple, we assumed in every simulation scenario of this study for the exogenous assets a value of 1 for each firm, 
i.e. $\bfa = (1\ld 1)^t\in\R^n$. 
For firm $i$'s nominal debt value, we set a fixed value $d_i=d$ for all $i\in\mathcal N$ and added a random variation to each debt value in order to
get differing setups which leads to 
\begin{equation}\label{eq:sim_debt_value}
\bfd = (d_1\ld d_n)^t + (\eps_1\ld \eps_n)^t \in\R^n,
\end{equation}
where the $\eps_i$ are independently normally distributed with mean value 0 and standard deviation 0.5, i.e. $\eps_i\sim N(0, 0.25)$. 
Note that in case of shocks with $\eps_i < d_i$ we set $d_i=0$ to avoid negative liabilities. 

When constructing an ownership matrix, the degree of ownership that is operationalized by the expression \emph{integration} 
can provide some crucial information. 

\begin{definition}
Consider a financial System $\mathcal F = (\bfa, \bfd, \bfMd, \bfMs)$. 
The \emph{debt integration level $\mu^{\bfd}$} is defined as the maximum column sum of $\bfMd$, i.e. 
\begin{equation}
\nu^{\bfd} = \max_{i\in\mathcal N} \sum_{j=1}^n M_{ij}^{\bfd} = \|\bfMd\|. 
\end{equation}
Analogously, 
\begin{equation}
\nu^{\bfs} = \max_{i\in\mathcal N} \sum_{j=1}^n M_{ij}^{\bfs} = \|\bfMs\| 
\end{equation}
is called the \emph{equity integration level}. 
\end{definition}

The integration level is hence a measure of the extent of cross-ownership in either the debt or the equity component 
and its definition is based on the one given in the work of \citet{elliott13}. 
Because of Assumption \ref{assu:os_mat_norm}, it follows directly that $\nu^{\bfd}, \nu^{\bfs}\in[0,1)$. 

The integration levels $\nu^{\bfd}$ and $\nu^{\bfs}$ obviously do not specify the single entries of the ownership matrices. 
For this purpose we will limit our consideration in the following to somewhat regular structures of the matrices.

\begin{definition}
An ownership matrix $\bfM$ is called 
\begin{itemize}
  \item a \emph{ring ownership matrix} if in every column only one entry is larger than 0 and 
  \item a \emph{complete ownership matrix} if every entry, except for the diagonal entry, is larger than 0 and of the same size.  
\end{itemize}
Further, let $\widetilde\bfM$ be a ring ownership matrix and $\widehat\bfM$ be a complete ownership matrix. 
A \emph{$\lambda$-convex combination} of $\widetilde\bfM$ and $\widehat\bfM$ is defined as the matrix $\bfM$ with entries 
\begin{equation}\label{eq:os_mat_convex}
M_{ij} = \lambda \widetilde M_{ij} + (1-\lambda) \widehat M_{ij}, \quad \lambda\in[0,1]. 
\end{equation}
\end{definition}

The concepts of ring and complete matrices and of convex combinations are originally used in \citet{acemoglu13}. 
If $\bfMd$ is a ring matrix this means that every firm has only one creditor within the system, 
and only one shareholder if $\bfMs$ is a ring matrix. 
Without loss of generality we assume that firm $i+1$ is the creditor (shareholder) of firm $i$ for $i=1\ld n-1$ 
and that firm $n$ is the creditor (shareholder) of firm 1. 
When $\bfMd$ ($\bfMs$) is a complete ownership matrix, debt (share) proportions are equally distributed between the $n-1$ firms. 
The lower $\lambda$ is chosen, the more equal are the entries of the corresponding convex combination. 

\begin{example}
For a system of size $n=4$ we assume a debt integration level of $\nu^{\bfd}=0.9$ and set $\lambda = 0.5$.
The ring ownership matrix $\widetilde\bfM^{\bfd}$, the complete ownership matrix $\widehat\bfM^{\bfd}$ 
and the $\lambda$-convex combination matrix $\bfMd$ then are
\begin{equation}
\widetilde\bfM^{\bfd} = \begin{pmatrix} 0  & 0  & 0  & .9 \\
                                        .9 & 0  & 0  & 0  \\
                                        0  & .9 & 0  & 0  \\
                                        0  & 0  & .9 & 0  \\\end{pmatrix},\ 
\widehat\bfM^{\bfd} = \begin{pmatrix} 0  & .3 & .3 & .3 \\
                                      .3 & 0  & .3 & .3 \\
                                      .3 & .3 & 0  & .3 \\
                                      .3 & .3 & .3 & 0  \\\end{pmatrix}\ \text{and}\ 
\bfMd = \begin{pmatrix}  0  & .15 & .15 & .6  \\
                        .6  & 0   & .15 & .15 \\
                        .15 & .6  & 0   & .15 \\
                        .15 & .15 & .6  & 0   \\\end{pmatrix}.
\end{equation}
\end{example} 

Consider two financial systems $\mathcal F = (\bfa,\bfd,\bfMd,\bfMs)$ and 
$\widetilde \F = (\widetilde\bfa,\widetilde\bfd,\widetilde\bfM^{\bfd},\widetilde\bfM^{\bfs})$ 
with corresponding integration levels $\nu^{\bfd}$, $\nu^{\bfs}$ and $\tilde\nu^{\bfd}$, $\tilde\nu^{\bfs}$. 
Due to the regular structure of the ownership matrices we can say that a system $\widetilde \F$ is more \emph{debt-integrated} than $\F$, 
if and only if $\tilde\nu^{\bfd}>\nu^{\bfd}$. 
In the same way, we define that a system is more \emph{equity-integrated}. 

With the former definitions, the following parameters are needed for the simulation of a financial system: $n$, $d$, $\nud$, $\nus$ and $\lambda$, 
where we will use the same $\lambda$ to define the debt and the equity ownership matrix according to \eqref{eq:os_mat_convex}. 
A \emph{simulated system} is the financial system $\F = (\bfa,\bfd,\bfMd,\bfMs)$, 
where the parameters $n$, $d$, $\nud$, $\nus$ and $\lambda$ are used to define $\bfa$, $\bfMd$ and $\bfMs$ 
and where the liabilities $\bfd$ are a realization of the random variable in \eqref{eq:sim_debt_value}.

\subsection{Effect of the Lag Value}\label{subsec:effect_lag_value}

As mentioned in Section \ref{subsec:trial_error_alg_decr}, the smaller the lag value $l$ is chosen in the Trial-and-Error Algorithms, 
the higher is the chance that the first possible default set is not the actual $D^*$. 
This results in unnecessary computation steps to reach the real default set. 
On the other hand, if $l$ is taken as very high, say $l=5$ or higher, there are many iteration steps in the algorithm that are possibly not needed. 
For this reason, we wanted to investigate this trade-off situation by determining the error rate for the first potential default set. 

Assume that for a given parameters $n$, $d$, $\nud$, $\nus$ and $\lambda$, we have generated $N$ simulated systems. 
For every system we determine for the Trial-and-Error Picard (TP), the Trial-and-Error Elsinger (TE) 
and the Trial-and-Error Hybrid Algorithm (TH) for a lag value $l\ge 2$ the first potential default set 
$\bar D^j_{{\rm TP}}(l)$, $\bar D^j_{{\rm TE}}(l)$ and $\bar D^j_{{\rm TH}}(l)$ where $j=1\ld N$. 
In case of the Trial-and-Error Picard Algorithm we define
\begin{equation}
\eps^j_{{\rm IP}}(l) = \begin{cases} 1, & \text{if $\bar D^j_{{\rm IP}}(l) \not= D^*$} \\ 
                                     0, & \text{else}, \end{cases}
\end{equation}
and analogously $\eps^j_{\rm TE}(l)$ and $\eps^j_{\rm TH}(l)$ for the TE and TH Algorithm, respectively. 
The \emph{error rate} for the TP Algorithm for the lag value $l$ is then given by
\begin{equation}
\eps_{\rm TP}(l) = \frac 1N \sum_{j=1}^N \eps^j_{\rm TP}(l) \in [0, 1].
\end{equation}
In the same way the error rates $\eps_{\rm TE}(l)$ and $\eps_{\rm TH}(l)$ are defined. 

For the investigation of the error rate we chose $d_i=d=1.5$ in \eqref{eq:sim_debt_value} as the debt value. 
The debt integration values where $\nu^{\bfd}\in\{0.9, 0.5, 0.1\}$, which we considered as systems with high, moderate and low debt cross-ownership. 
Similarly, we took $\nu^{\bfs}\in\{0.45, 0.25, 0.05\}$ for equity integration, where each value is half the associated debt integration. 
The justification for this approach is that equity cross-ownership is probably commonly less pronounced than debt cross-ownership. 
Further, we wanted to avoid possible cross-ownership entries lager that 0.5 
since this would mean that a firm is owned by majority by another firm in the system. 
Each equity and debt integration value was combined with each other which results in 9 possible system settings. 
Beyond that, all 9 settings were investigated three times, 
where the structure parameter $\lambda$ took the three possible values $\lambda\in\{0, 0.5, 1\}$,
i.e. systems with only ring ownership matrices, systems with only complete matrices and systems with a 0.5-convex combination were considered. 
In total, the combination of the parameters $\nud$, $\nus$ and $\lambda$ leads to 27 different settings 
and for every setting $N=1000$ simulated systems were generated.
The error rates were calculated for the three algorithms for the lag values $l=2\ld 7$. 
Repetitive simulations with $N=1000$ showed that the error rates are fairly stable for different simulation runs 
which is why we viewed the number of 1000 repetitions as reliable. 
We used the Decreasing Trial-and-Error Algorithm defined in Algorithm \ref{alg:trial_error_decr}, 
but simulations with the Increasing Trial-and-Error Algorithm showed very similar results. 

\begin{table}[!htbp]
\centering
\caption{Error rates $\eps^k_{\rm TP}(l)$, $\eps^k_{\rm TE}(l)$ and $\eps^k_{\rm TH}(l)$ 
in percentage points for the Decreasing TP, TE and TH Algorithm for $l=2\ld 7$. 
Mean values over all three values of $\lambda$ and all values of $\nu^{\bfd}$ and $\nu^{\bfs}$ are shown for each combination.
The last three rows of the table show the overall mean error rates over all considered system sizes.}
\label{tab:lag_sim_res_n}
\begin{tabular}{llrrrrrr}\toprule
                         &    & $l=2$  & $l=3$ & $l=4$ & $l=5$ & $l=6$ & $l=7$ \\ \midrule
\multirow{3}{*}{$n=5$}   & TP & 9.652  & 3.111 & 1.170 & 0.500 & 0.230 & 0.104\\
                         & TE & 3.452  & 0.718 & 0.159 & 0.037 & 0.022 & 0.007\\
                         & TH & 0.156  & 0.015 & 0.004 &     0 &     0 &     0\\ \midrule                         
\multirow{3}{*}{$n=10$}  & TP & 4.574  & 1.452 & 0.529 & 0.196 & 0.089 & 0.048\\
                         & TE & 1.256  & 0.111 & 0.015 & 0.004 &     0 &     0\\
                         & TH & 0.233  & 0.004 &     0 &     0 &     0 &     0\\ \midrule
\multirow{3}{*}{$n=25$}  & TP & 8.608  & 2.730 & 1.063 & 0.455 & 0.207 & 0.067\\
                         & TE & 2.819  & 0.348 & 0.041 & 0.004 &     0 &     0\\
                         & TH & 0.533  &     0 &     0 &     0 &     0 &     0\\ \midrule
\multirow{3}{*}{$n=50$}  & TP & 10.596 & 3.122 & 1.111 & 0.482 & 0.189 & 0.078\\
                         & TE & 3.300  & 0.389 & 0.041 & 0.007 &     0 &     0\\
                         & TH & 0.537  & 0.011 &     0 &     0 &     0 &     0\\ \midrule
\multirow{3}{*}{$n=100$} & TP & 11.359 & 3.385 & 1.211 & 0.478 &  0.2  & 0.093\\
                         & TE & 3.233  & 0.411 & 0.044 & 0.007 &  0    &     0\\
                         & TH & 0.404  & 0.007 &     0 &     0 &  0    &     0\\ \midrule
\multirow{3}{*}{$n=200$} & TP & 11.485 & 3.230 & 1.111 & 0.419 & 0.167 & 0.063\\
                         & TE & 2.870  & 0.322 & 0.022 & 0.007 &     0 &     0\\
                         & TH & 0.267  & 0.007 &     0 &     0 &     0 &     0\\ \midrule
\multirow{3}{*}{all $n$} & TP & 9.379  & 2.838 & 1.033 & 0.422 & 0.180 & 0.075\\
                         & TE & 2.822  & 0.383 & 0.054 & 0.011 & 0.004 & 0.001\\
                         & TH & 0.355  & 0.007 & 0.001 &     0 &     0 &    0\\ \bottomrule
\end{tabular}
\end{table}

The first observation was that the structure of the ownership matrix, that is the choice of $\lambda$, had no severe effect on the error rates. 
This results in comparing the error rates for systems with given $n$, $\nu^{\bfd}$ and $\nu^{\bfs}$ for the three choices of $\lambda$. 
In case of $n=100$ where moderate debt and low equity cross-ownership was present ($\nu^{\bfd}=0.5$, $\nu^{\bfs}=0.05$) and a lag value $l=1$, 
the largest absolute difference is documented for the Trial-and-Error Picard Algorithm 
between complete and ring ownership matrices (9.6\% to 3.2\%). 
In the large majority of possible combinations, the difference was much smaller which is why we concluded that the ownership structure itself 
does not affect the error rate in an essential way. 
For this reason, we summarized the three values of $\lambda$ and calculated the mean error rates for every combination of $n$, $\nu^{\bfd}$ 
and $\nu^{\bfs}$ over all $\lambda$ for the further results. 

In Table \ref{tab:lag_sim_res_n}, the results of the simulation to investigate the effect of the system size are summarized. 
The error rates are calculated as the mean over all possible combinations for debt and equity integration for every value of $n$. 
We observe that the network size only slightly affects the error rates, since for the same lag value they are relatively close for all system sizes.
The only exception are systems with $n=10$ firms, where the error rate is smaller compared to the other systems. 
In the last three rows of the table, an overall impression of the error rates shows that even for the TP Algorithm 
and a lag value of $l=2$, the error rate is not higher than 10\% in total. 
We also detected that the error rates for the TE Algorithm are much smaller and even more so for the TH Algorithm. 
For increasing lag values, the error rates quickly diminish in size for all considered methods. 

\begin{table}[!htbp]
\centering
\caption{Error rates $\eps^k_{\rm TP}(l)$, $\eps^k_{\rm TE}(l)$ and $\eps^k_{\rm TH}(l)$ 
in percentage points for the TP, TE and TH Algorithm for $l=2,3$. 
For the debt integration level $\nu^{\bfd}$, low, moderate and high integration is defined as 0.1, 0.5 and 0.9, respectively. 
For the equity integration level $\nu^{\bfs}$ the corresponding levels are defined as 0.05, 0.25 and 0.45. 
Mean values over all three values of $\lambda$ and all values of $n$ are given for each combination. 
Further, overall mean error rates for within each debt and equity integration level are shown in an additional column and row.}
\label{tab:lag_sim_res_alg}
\begin{tabular}{llrrrrclrrrr}\toprule
 & \multicolumn{5}{c}{$l=2$} & & \multicolumn{5}{c}{$l=3$}\\\cmidrule{2-6}\cmidrule{8-12}
\multirow{6}{*}{TP} & $\nu^{\bfd}$ & \multicolumn{3}{c}{$\nu^{\bfs}$} & & & $\nu^{\bfd}$ & \multicolumn{3}{c}{$\nu^{\bfs}$} & \\ \cmidrule{3-5} \cmidrule{9-11}
 &      & low  & mod.  & high  &       & &      & low  & mod. & high  &     \\
 & low  & 2.60 & 4.65  & 7.64  & 4.96  & & low  & 0.22 & 0.54 & 1.41  & 0.72\\
 & mod. & 8.53 & 19.64 & 35.83 & 21.33 & & mod. & 1.72 & 5.58 & 14.62 & 7.31\\
 & high & 1.31 &  2.47 & 1.75  & 1.84  & & high & 0.14 & 0.73 &  0.57 & 0.48\\ \cmidrule{2-6}\cmidrule{8-12} 
 &      & 4.11 & 8.73  & 15.14 & 9.38  & &      & 0.69 & 2.29 &  5.53 & 2.84\\ \midrule
\multirow{6}{*}{TE} & $\nu^{\bfd}$ & \multicolumn{3}{c}{$\nu^{\bfs}$} & & & $\nu^{\bfd}$ & \multicolumn{3}{c}{$\nu^{\bfs}$} & \\ \cmidrule{3-5} \cmidrule{9-11}
 &      & low  & mod.  & high  &       & &      & low  & mod. & high  &     \\
 & low  & 2.16 & 2.34  & 2.51  & 2.34  & & low  & 0.19 & 0.19 & 0.18  & 0.19\\
 & mod. & 6.24 & 6.07  & 5.60  & 5.97  & & mod. & 1.07 & 0.91 & 0.86  & 0.95\\
 & high & 0.29 & 0.13  & 0.06  & 0.16  & & high & 0.03 & 0.02 & 0     & 0.02\\ \cmidrule{2-6}\cmidrule{8-12}
 &      & 2.90 & 2.84  & 2.72  & 2.82  & &      & 0.43 & 0.37 & 0.35  & 0.38\\ \midrule
\multirow{6}{*}{TH} & $\nu^{\bfd}$ & \multicolumn{3}{c}{$\nu^{\bfs}$} & & & $\nu^{\bfd}$ & \multicolumn{3}{c}{$\nu^{\bfs}$} & \\ \cmidrule{3-5} \cmidrule{9-11}
 &      & low  & mod.  & high  &       & &      & low  & mod. & high  &     \\
 & low  & 0.02 & 0.13  & 0.24  & 0.13  & & low  &  0   & 0    & 0     & 0   \\
 & mod. & 0.22 & 0.96  & 1.59  & 0.92  & & mod. &  0   & 0.01 & 0.04  & 0.02\\
 & high & 0.04 & 0     & 0     & 0.01  & & high &  0   & 0    & 0     & 0   \\ \cmidrule{2-6}\cmidrule{8-12}
 &      & 0.09 & 0.36  & 0.61  & 0.35  & &      &  0   & 0.01 & 0.01  & 0.01\\ \bottomrule 
\end{tabular}
\end{table}

To assess the influence of the debt and the equity integration level, the mean error rate over all considered system sizes are taken 
and listed for each possible combination of the integration level, as shown in Table \ref{tab:lag_sim_res_alg}. 
If the debt integration level increases from low to moderate, the error rates increase as well. 
This can, for example, be seen when comparing the mean error rates for the debt integration levels 0.1 and 0.5 
over all equity integration levels in the last column in Table \ref{tab:lag_sim_res_alg} for every lag value. 
For the Trial-and-Error Picard Algorithm, the error rate increases from 4.96\% to 21.33\%; for the other algorithms we observe similar results. 
A further increase of the debt integration from 0.5 to 0.9, however, has the reverse effect since the error rates decrease in this case. 
Again we take the TP Algorithm as an example where the error rates shrinks from 21.33\% to 1.84\%. 
A possible explanation for this behavior could be that for the TP Algorithm the number of needed iteration steps to converge to the solution 
was always highest for the combination $\nu^{\bfd}=0.5$ and $\nu^{\bfs}=0.45$ (data not shown). 
Hence, the convergence speed in these situations is very slow 
which explains why the first potential default set is often not the actual default set. 
When the mean error rates for every equity integration level averaged over all values of $\nu^{\bfd}$ is examined, 
we observed that, except for the TE Algorithm, the error rates increase for increasing integration levels. 
In case of the TP Algorithm, we have an increase from 4.11\% to 8.73\% to 15.14\% 
for $\nu^{\bfs}=0.05$, $\nu^{\bfs}=0.25$ and $\nu^{\bfs}=0.45$, respectively. 
The error rates for the TE Algorithm stay approximately constant (2.90\%, 2.84\%, 2.72\%). 

Our overall conclusion of this part of the simulation study is that the choice of the algorithm 
and the lag value $l$ has the strongest effect on the error rate, i.e the error rates quickly decrease for greater lag values. 
The TP Algorithm has, as expected, the highest overall error rates, much higher than the TE and the TH Algorithm. 
What also affects the error rate is the integration level of the ownership matrices. 
For our simulation setting, it was the combination of moderate debt and high equity integration that yielded in the highest rates. 
The structure of the ownership matrices, on the other hand, had no influence on the error rate. 
Taking the overall mean error rates as the main reference, we can state that for the TE and the TH Algorithm a lag value of 2 is appropriate, 
since the corresponding error rates are with 2.82\% and 0.36\% very small. 
For the TP Algorithm, however, for $l=2$ we get an overall error rate of 9.38\% 
which is why a lag value of 3 with an overall error rate of 2.84\% seems more convenient for this procedure. 

\subsection{Comparison of Algorithm Efficiency}\label{subsec:algo_efficency}

Searching for the most efficient algorithm, the main issue is to minimize the calculation effort to find the solution $\bfR^*$. 
In every iteration step of the algorithm, different kind of calculations are carried out for the different algorithms. 
We quantify the calculation costs with the Landau symbol (Big $\Omikron$ notation), 
where for example $\Omikron(n)$ means that the time $T(n)$ to compute a problem of size $n$ grows at the rate $n$. 
We distinguish between two different types of calculations in our considerations. 
For the first type, a mapping is applied to a given vector. 
This mapping can either be the mapping $\Phi$ in \eqref{eq:Phi}  
or the mappings $\Phi^{\bfd}$ and $\Phi^{\bfs}$ in \eqref{eq:phi_aux_eq} and \eqref{eq:phi_aux_debt}, respectively, 
whereas the second type, matrix multiplications are done. 
In all cases, the most expensive calculations are matrix multiplications, 
whereas the type embodies the solution of a linear equation system such as the ones defined in the Algorithms \ref{alg:equity} and \ref{alg:debt}. 
The computational costs of both types are between the range of $\Omikron(n^2)$ and $\Omikron(n^3)$ (cf. \citet{dahlquist08}).

Keeping the functioning of the algorithms in mind, the Elsinger and the Hybrid Algorithm seem to be less efficient than the Picard Algorithm, 
since in the latter one, no linear equation system has to be solved which results in smaller computational costs. 
However, as we have seen in the Propositions \ref{prop:elsinger_better_picard} and \ref{prop:hybrid_better_elsinger}, 
in terms of iteration steps, the Hybrid Algorithm converges faster to the solution compared to the Elsinger Algorithm, 
which in turn converges faster to the solution than the Picard Iteration. 
Therefore, a typical trade-off-situation is given between computational costs and convergence speed of an algorithm. 
Note that for a sequence $\bfR^k$ that converges to a fixed point $\bfR^*$ the \emph{convergence rate} is called \emph{linear}, 
if there exists a $c\in(0,1)$ such that 
\begin{equation}
\|\bfR^* - \bfR^{k+1}\| \le c \|\bfR^* - \bfR^k\|
\end{equation}
for all $k\ge0$. 
Since $\|\Phi(\bfR^*)-\Phi(\bfR^k)\|\le\Imax\|\bfR^*-\bfR^k\|$ with $\Imax=\max\{||\bfMd||,||\bfMs||\}$ 
(see Lemma 4.1 in \citet{fischer14}), 
linear convergence holds for the Picard Algorithm if instead of Assumption \ref{assu:holding_mat}
the stronger assumption of matrix norms being strictly smaller than $1$ is made (which means that of all debt
and equity a non-zero share is held by a system outsider).
The properties of the Elsinger and the Hybrid Algorithms, however, 
made it impossible to prove linear convergence (or an even higher convergence rate). 
The next problem is that the total computational cost for one the algorithms is not determinable in general. 
For these reasons, the comparison of the different calculation techniques (Picard, Elsinger, Hybrid) on an analytical base seems impossible. 

This is why we measured the time that was needed to execute an algorithm and considered this value as the primary outcome of our simulation. 
Though this measure strongly depends on the processor speed and memory capacity of the computer, 
it allows an objective comparison of the different algorithms. 
The simulations were conducted on a computer with 3.2 GHz and 4 GB RAM, the software used was R (\citet{Rsoftware}).

The parameters defining the financial systems are given in the following. 
Unlike to the simulation in Section \ref{subsec:effect_lag_value}, where a fixed debt value $d$ was used, 
we varied between five possible the debt values and chose $d\in\{1, 1.5, 2, 2.5, 3\}$. 
The set of equity and debt integration levels was extended to $\nu^{\bfs}\in\{0.025,0.1,0.175,0.25,0.325,0.4,0.475\}$ and 
$\nu^{\bfd} \in \{0.05,0.2,0.35,0.5,0.65,0.8,0.95\}$, hence seven possible integration level respectively. 
A result in Section \ref{subsec:effect_lag_value} was that the structure of the ownership matrix does not influence the error rates, 
which is why we only took complete ownership matrices into account for this simulation, i.e. $\lambda=0$ in \eqref{eq:os_mat_convex}. 
Together with the six considered systems sizes ($n\in\{5, 10, 25, 50, 100, 200\}$), this new setting results in 
$6\cdot5\cdot7\cdot7=1470$ different settings. 
Again, for each setting, $N=1000$ simulated systems were generated for every parameter combination. 
For every simulated systems we applied all 15 algorithms presented in the former section 
and documented the runtime for every procedure to find the solution $\bfR^*$. 
We used the Picard, the Elsinger and the Hybrid Algorithm (Algorithms \ref{alg:picard}, \ref{alg:elsinger} and \ref{alg:comb_method}) 
with both versions, i.e. the decreasing and the increasing version. 
Further, the Trial-and-Error versions were considered, again both the decreasing (Algorithm \ref{alg:trial_error_decr}) 
and the increasing version (Algorithm \ref{alg:trial_error_incr}) with lag values of $l=3$ for the Picard versions 
and $l=2$ for the Elsinger and the Hybrid versions. 
The choice of the lag value is a result of the simulations in Section \ref{subsec:effect_lag_value}. 
Note that minimizing the error rate to an appropriate value is not necessarily equivalent to minimizing the runtime of the algorithms. 
For these reasons, we compared for every considered scenario the runtime of the Trial-and-Error Algorithms using lag values from $l=2$ to $l=5$. 
The results (not shown here) are that for the Elsinger and the Hybrid versions of the algorithms, 
the choice of $l=2$ does not only keep the error rate on a very low level, but also minimizes the runtime. 
For the Trial-and-Error Picard Algorithm, the simulation showed that the runtime is almost identical for $l=2$
and $l=3$. 
However, there is no clear tendency between those choices of $l$: 
for some parameter combinations $l=2$ lead to smaller runtimes and for some situations this was the case for $l=3$. 
Due to this indifference for the Trial-and-Error Picard Algorithm, 
we set $l=3$ for these procedures in accordance with the findings of Section \ref{subsec:effect_lag_value}. 
At last, the three versions of the Sandwich Algorithm (Algorithm \ref{alg:sandwich}) were also taken into account. 
The tolerance level in all algorithms was set to $\varepsilon=10^{-3}$. 

\pagebreak

\begin{table}[!htbp]
\centering
\caption{Mean runtime in seconds for every algorithm over all debt and equity integration values ($\nu^{\bfd}$ 
and $\nu^{\bfd}$) and all debt values ($d$) grouped by system size, algorithm and iteration type.
For each of the three iteration types, the average runtime of the corresponding increasing and decreasing version was calculated, 
except for the Sandwich Algorithms. 
}
\label{tab:res_sim_size}
\begin{tabular}{llrrrrrr}\toprule
\multirow{2}{*}{Algorithm Type}  &  \multirow{2}{*}{Iteration Type} & \multicolumn{6}{c}{system size $n$} \\\cline{3-8}
                                 &          &    5 &   10 &   25 &   50 &   100 &    200 \\\midrule
\multirow{3}{*}{Non-finite}      & Picard   & 1.81 & 1.90 & 2.27 & 3.01 &  6.19 &  21.16 \\
                                 & Elsinger & 1.86 & 2.08 & 2.90 & 5.71 & 24.38 & 175.40 \\
                                 & Hybrid   & 1.65 & 1.84 & 2.58 & 5.24 & 23.06 & 164.03 \\\cline{2-8}
                                 &          & 1.78 & 1.94 & 2.58 & 4.65 & 17.87 & 120.20 \\\midrule
\multirow{3}{*}{Trial-and-Error} & Picard   & 1.29 & 1.30 & 1.64 & 2.81 &  8.77 &  45.21 \\
                                 & Elsinger & 1.32 & 1.39 & 1.91 & 3.95 & 17.14 & 119.38 \\
                                 & Hybrid   & 1.57 & 1.67 & 2.24 & 4.49 & 19.04 & 130.88 \\\cline{2-8}
                                 &          & 1.39 & 1.45 & 1.93 & 3.75 & 14.98 &  98.49 \\\midrule
\multirow{3}{*}{Sandwich}        & Picard   & 1.16 & 1.28 & 1.79 & 3.31 & 10.73 &  55.08 \\
                                 & Elsinger & 1.33 & 1.54 & 2.41 & 5.53 & 26.18 & 192.65 \\
                                 & Hybrid   & 1.49 & 1.70 & 2.55 & 5.57 & 25.44 & 182.91 \\\cline{2-8}
                                 &          & 1.32 & 1.51 & 2.25 & 4.80 & 20.78 & 143.55 \\\midrule
\multirow{3}{*}{Overall}         & Picard   & 1.47 & 1.54 & 1.92 & 2.99 &  8.13 &  37.56 \\
                                 & Elsinger & 1.54 & 1.69 & 2.41 & 4.97 & 21.84 & 156.44 \\
                                 & Hybrid   & 1.59 & 1.74 & 2.44 & 5.01 & 21.93 & 154.55 \\\cline{2-8}
                                 &          & 1.53 & 1.66 & 2.26 & 4.32 & 17.30 & 116.19 \\\bottomrule
\end{tabular}
\end{table}

An important topic is of course the comparison of new developed techniques (Trial-and-Error, Sandwich) with the existing procedure. 
In Table \ref{tab:res_sim_size}, the mean runtimes are listed grouped by the size of the financial system 
as well as the algorithm and the iteration type. 
We summarized the decreasing and the increasing version of every algorithm respectively by calculating the mean runtime of both procedures. 
The runtimes of the decreasing versions were in most situations smaller than their counterparts. 
Ignoring the random structure of $\bfd$ for an instance and calculating the fixed point, it could be seen
that in about 60~\% of all considered scenarios, 
no firm was in default which explains the slight `overperformance' of the decreasing algorithms. 
If we compare the mean runtimes over all iteration types, we find that for $n=5$ the Sandwich Algorithms have the best performance (1.32~s) 
compared to the Trial-and-Error and the non-finite methods (1.39~s and 1.78~s, respectively). 
For $n\ge5$, the fastest runtimes averaged over the iteration types are achieved for the Trial-and-Error procedures. 
Comparing the different iteration techniques with each other, 
we find that using the Picard Iteration technique results in a minimal computational effort for all considered system sizes. 
To be more specific, for small financial systems ($n=5,10$), the Sandwich Picard Algorithm shows the smallest runtime 
(1.16 s and 1.28 s, respectively), whereas for medium-sized systems ($n=25,50$), 
the Trial-and-Error Picard Algorithm performs best compared to the other algorithm types (1.64 s and 2.81 s, respectively). 
In case of large financial systems, i.e. $n=100,200$, the Picard Iteration in its non-finite form yields to lowest runtimes 
(6.19 s and 21.16 s, respectively). 
In general, it is clearly visible that Picard-typed algorithms have the best performance within every Algorithm type. 
The only exception of this trend can be found in the class of non-finite algorithm, where for $n=5$ and $n=10$ 
the Hybrid Algorithm showed a slightly lower runtime than the Picard Algorithm. 
In all other situations, however, the Picard Algorithm is superior to the other algorithms. 

Beside the size $n$, we also investigated the influence of the other parameters that define the form of the financial system on the runtime. 
We observed that increasing debt values $d$ result in an increasing calculation effort, 
see Table \ref{tab:res_sim_debt_int} in the Appendix for a detailed overview. 
An exception of this tendency represents the Hybrid Algorithm, where the runtime for large $d$ begins to decrease again, 
no matter which algorithm type is considered.  
The reason for this behavior is that the runtime for the increasing versions of the Hybrid Algorithm becomes smaller for large $n$ 
and so does the average of the increasing and the decreasing version of the algorithm, that is shown in Table \ref{tab:res_sim_debt_int}. 
A possible explanation is that the Increasing Hybrid Algorithm uses a Picard-type technique to determine the next debt iterate 
(cf. Algorithm \ref{alg:picard_debt}). 
As shown above, the Picard iteration results, in particular for large $n$, in much better runtime performances. 

If the debt integration level increases, we first observe a similar effect on the runtime as for the debt values, 
i.e. the higher the integration level, the higher the runtime. 
However, this monotonicity holds only up to $\nu^{\bfd}=0.5$ or $\nu^{\bfd}=0.65$ in most cases. 
For larger debt integration levels, the computational effort decreases again. 
The reason for this behavior could be that for small values of $\nu^{\bfd}$,  
it is very likely that many firms in the financial system are in default. 
In such situations, we observe that only few iteration steps are needed until $\bfR^*$ is reached. 
If the debt integration level is very high, the same effect establishes with the difference that many firms in the system are solvent. 
For medium debt integration levels this clear distinction for a firm between solvent and default disappears. 
A consequence is a higher number of needed iteration steps which also influences the runtime. 
Moreover, this interpretation is underlined by the fact that for small $\nu^{\bfd}$ (firms more likely in default), 
the increasing versions of the algorithms have a better performances, 
whereas this relationship inverts for large integration levels where the firms are more likely to be solvent.
For increasing equity integration levels this effect is not visible, the runtime increases if $\nu^{\bfs}$ increases (results not shown). 
Since $\nu^{\bfs}\le 0.45$ the equity integration seems to have a less strong effect on the status of the firms in the system 
and therefore, the effect seen in the debt integration levels does probably not appear.

\section{Summary}\label{sec:summary}

In this article, we gave a survey of the existing algorithms (``Picard'' and ``Elsinger'') for the computation
of equilibrium prices in a financial system in which cross-holdings of equity and debt are present. 
Moreover, we showed how the ideas of \citet{elsinger09} and \citet{eisenberg01} can be combined to get an iteration procedure (``Hybrid Algorithm'')
that is in every iteration step closer to the solution $\bfR^*$ than the ``Picard'' and ``Elsinger'' algorithms. 
We developed new iteration methods based on the information of defaulting and solvent firms under a current payment vector. 
A consequence of these default set-based methods is that the exact solution of the system is reached in a finite number of steps, 
which could not be ensured for the existing iteration procedures. 
Using this new approach yields to two different concepts, that we called ``Sandwich'' algorithms and ``Trial-and-Error'' algorithms. 
While for the former type, a clear stopping criteria can be defined (at least almost surely), 
the latter algorithms have the drawback that every potential solution has to be checked for validity. 
In a simulation study, we showed that choosing an appropriate lag value $l$, the computational effort can be kept to a minimum. 

Another simulation showed that essentially less iteration steps have to be performed when using the new default-set based techniques. 
However, the faster convergence concerning the number of iterations has its price: 
In the new algorithms other than the Picard type, potentially several linear equation systems have to be solved in every iteration step. 
This leads to a higher calculation effort for those methods and a result of the empirical investigation of the runtime for all algorithms was 
that in particular for large financial systems the computational costs then become  higher than for algorithms where no linear equation systems have to be solved.
Another result of the runtime analysis is that the most efficient iteration technique is of Picard type. 
In the majority of the considered settings those iteration techniques performed best, 
no matter which algorithm type (Non-finite, Trial-and-Error, Sandwich) was used. 
One of the main results is that the choice of the most efficient algorithm strongly depends on the size $n$ of the financial system. 
We observed that for small systems ($n=5,10$) the Sandwich Picard technique, for medium-sized systems ($n=25,50$) the Trial-and-Error Picard technique 
and for large systems ($n=100,200$) the simple non-finite Picard technique achieves the best results with regard to the minimization of the runtime. 
Regarding the choice of the tolerance level $\eps$, smaller values than the used one of $\eps=10^{-3}$ will strongly affect 
the results of the non-finite iteration techniques, since additional simulations (results not listed here) revealed that an increase of $\eps$ will lead to a disproportionally strong increase of the 
needed iteration steps and therefore the runtime. 
One consequence could be that for large systems the non-finite Picard iteration would not be optimal anymore,
since the finite algorithm techniques do not depend on $\eps$. 
We are aware of this effect, but we think that, for practical purposes, a tolerance level of $\eps=10^{-3}$ is sufficiently small.

The simulation in Section \ref{subsec:algo_efficency} contained only complete ownership matrices ($\lambda = 0$ in \eqref{eq:os_mat_convex}). 
It is of potential interest, whether in case of ring ownership matrices or $\lambda$-convex combinations 
the results lead to the same conclusions. 
No matter which value of $\lambda$ is chosen, the entries of $\bfMd$ and $\bfMs$ still are uniquely determinable. 
A potential extension of this assumption would be to allow random ownership matrices based on a random network matrix 
as used for example in \citet{elliott13}. 
Besides these questions, the main focus for further research should be on generalizing 
the algorithms for systems with more than one seniority level for the liabilities. 
In \citet{fischer14}, this was done for the non-finite Picard Algorithm 
and in \citet{elsinger09}, an extension of the non-finite Elsinger Algorithm is discussed as well. 
It would be of interest whether and -- if yes -- how the Hybrid Algorithm can be generalized for a model that allows for a seniority structure of debt.

\appendix

\setcounter{lemma}{0}
\renewcommand{\thelemma}{\Alph{section}\arabic{lemma}}

\section{Appendix}

\subsection{Proofs and Auxiliary Results}

\begin{lemma}
\label{non-expansive_infinite}
Let $||\cdot||$ be a not necessarily strictly convex norm on $\mathbb{R}^n$, and
let $\Phi$ be a map on a nonempty convex and compact set $C\subset \mathbb{R}^n$
which is non-expansive with respect to the norm-induced metric.
The set of fixed points of $\Phi$ in $C$ is
then nonempty, closed, and either a singleton, or uncountable.
\end{lemma}

\begin{proof}
Much-refined versions of this result are known (e.g.~\citet{bruck73}). 
For convenience, a short proof is given.
Non-expansiveness implies that $\Phi$ is (1-Lipschitz) continuous. 
The set of fixed points is hence closed, and the Brouwer--Schauder Fixed Point Theorem (e.g.~\citet{rudin91}) 
provides the existence of at least one fixed point.
Assume now that $\mathbf{x}, \mathbf{y}\in C$ are two distinct fixed points of $\Phi$.
For $\mathbf{v}\in C$ and $\varepsilon > 0$, 
$B_{\varepsilon}(\mathbf{v})
=
\{\mathbf{w}\in C: ||\mathbf{w}-\mathbf{v}|| \leq \varepsilon\}$
is a non-empty, convex and compact subset of $C$. For $\lambda\in(0,1)$, the intersection
\begin{equation}
C_{\lambda} = B_{\lambda||\mathbf{y}-\mathbf{x}||}(\mathbf{x}) 
\cap B_{(1-\lambda)||\mathbf{y}-\mathbf{x}||}(\mathbf{y})
\end{equation}
is non-empty (as it contains $(1-\lambda)\mathbf{x}+\lambda\mathbf{y}$)), convex 
and compact, and it contains neither $\mathbf{x}$, nor $\mathbf{y}$.
By the triangle inequality, $C_{\lambda_1}\cap C_{\lambda_2}=\emptyset$ for $\lambda_1 \neq \lambda_2$.
Non-expansiveness implies that $\Phi(C_{\lambda})\subset C_{\lambda}$.
By Brouwer--Schauder, there
exists a fixed point of $\Phi$ in $C_{\lambda}$. Hence there exist uncountably many
fixed points of $\Phi$ in $C$.
\end{proof}

\begin{lemma}
Let $\bfM\in\R^{n\times n}$ be an ownership matrix that has the Elsinger Property. 
Then $\rho(\bfM) < 1$, where
\begin{equation}
\rho(\bfM) = \max\{|\lambda_i|: \text{$\lambda_i$ eigenvalue of $\bfM$}\}
\end{equation}
is the spectral radius of $\bfM$.
\end{lemma}

\begin{proof}
A well known result (cf. \cite{rudin91}) is that $\rho(\bfM)\le\|\bfM\|\le 1$. 
In case of $\|\bfM\|<1$ there is nothing to show, so we assume that $\|\bfM\|=1$ 
which is no contradiction to the Elsinger Property of $\bfM$.  
We will show the claim by contradiction. 
To this end, assume that $\rho(\bfM)=1$. 
For the corresponding eigenvalue $\bfv$ is must hold that $\bfv\not=\bzero$ and $\bfM\bfv=\rho(\bfM)\bfv=\bfv$. 
We can formulate this equation alternatively as
\begin{equation}\label{eq:ev_zero}
(\bfI_n-\bfM)\bfv = \bzero_n.
\end{equation}
Since $\bfM$ has the Elsinger Property, it follows by \citet{elsinger09}, Lemma 1, that $(\bfI_n-\bfM)$ is invertible. 
But that means that there exists no vector $\bfv\not=\bzero$ such that \eqref{eq:ev_zero} is true. 
Hence, $\bfv=\bzero$ which is a contradiction and from which follows that $\rho(\bfM)<1$.
\end{proof}

\begin{lemma}\label{lem:inv_xos-mat}
Let $\bfM\in\R^{n\times n}$ be an ownership matrix that has the Elsinger property and for which $\rho(\bfM)<1$.  
Then $(\bfI_n - \bfM)^{-1}$ exists and can be obtained via the \emph{Neumann expansion}:
\begin{equation}
(\bfI_n - \bfM)^{-1} = \sum_{n=0}^{\infty} \bfM^n,
\end{equation}
where $\bfM^0 = \bfI_n$. 
Consequently, the diagonal entries of $(\bfI_n - \bfM)^{-1}$ are greater than or equal to 1 
and the other entries are all non-negative. 
\end{lemma}

\begin{proof}
See \cite{rudin91}. 
\end{proof}

\noindent\emph{Proof} of \textbf{THEOREM \ref{theo:unique_fp}}:
A proof of Theorem \ref{theo:unique_fp} is necessary because related proofs in \citet{suzuki02}, \citet{gourieroux12} and \citet{fischer14} 
rely on stronger matrix conditions than the Elsinger Property, while \citet{elsinger09} considers an
equation system which slightly differs from \eqref{eq:liq_eq_debt} and \eqref{eq:liq_eq_equity}. 
First, note that \eqref{eq:liq_eq_debt} and \eqref{eq:liq_eq_equity} can only have non-negative solutions. 
This is shown in Lemma 3.5 of \citet{fischer14} under stricter matrix conditions, but because of
Lemma \ref{lem:inv_xos-mat} of this paper, it is straightforward to see that the proof works in the same manner under the Elsinger Property. 
The interval $[\Rsmall,\Rgreat]$ is convex and compact, and $\Phi(\bfR)$ is continuous in $\bfR$. 
By Lemma \ref{lem:Phi_self-mapping} and the Brouwer-Schauder Fixed Point Theorem, it follows that at least one solution exists.
Furthermore, $\Phi$ as in Eq.~\eqref{eq:Phi} is a non-expansive mapping. 
This follows from Lemma 4.1 of \citet{fischer14}, where a strict contraction property is shown under stricter matrix conditions, 
but again it is straightforward to see how the corresponding proof implies non-expansiveness under the Elsinger Property for all ownership matrices. 
Since it follows from Proposition \ref{prop:def_set_sol} that there can be a maximum of $2^n$ possible solutions of 
\eqref{eq:liq_eq_debt} and \eqref{eq:liq_eq_equity}, uniqueness follows from Lemma \ref{lem:Phi_self-mapping} 
and Lemma \ref{non-expansive_infinite} in the Appendix.
\qed

\begin{lemma}\label{lem:inv_aux}
Let $\bfM\in\R^{n\times n}$ be an ownership matrix as in Lemma \ref{lem:inv_xos-mat}, such that $\mathcal N_0\subset \mathcal N$, and 
the matrix $\bfL\in\R^{n\times n}$ be defined as
\begin{equation}
\left(\bfL\right)_{ij} = \begin{cases} 1 ,& \text{if $i=j$ and $i\in \mathcal N_0$,} \\
                                       0 ,& \text{else.} \end{cases}               
\end{equation}
Then it holds that
\begin{equation}
(\bfI_n - \bfL\bfM\bfL)^{-1} \bfL \le \bfL(\bfI_n - \bfM)^{-1}\bfL.
\end{equation}
\end{lemma}

\begin{proof} 
Note that 
\begin{equation}
(\bfI_n-\bfL)^k = (\bfI_n-\bfL) \quad \text{for $k\in\N$}
\end{equation} 
and that $\bfM^0 =(\bfI_n-\bfL)^0 = \bfI_n$. 
Using Lemma \ref{lem:inv_xos-mat} we have that
\begin{equation}
\begin{split}
(\bfI_n - \bfL\bfM\bfL)^{-1} \bfL &= \left(\sum_{n=0}^{\infty}(\bfL\bfM\bfL)^n\right)\bfL \\
 &= (\bfI_n + \bfL\bfM\bfL + \bfL\bfM\bfL\bfM\bfL + \bfL\bfM\bfL\bfM\bfL\bfM\bfL + \ldots)\bfL \\
 &= \bfL + \bfL(\bfM + \underbrace{\bfM\bfL\bfM}_{\le \bfM^2} + \underbrace{\bfM\bfL\bfM\bfL\bfM}_{\le \bfM^3} + \ldots)\bfL \\
 &\le \bfL + \bfL\left(\sum_{n=1}^{\infty}\bfM^n\right)\bfL \\
 &= \bfL\left(\sum_{n=0}^{\infty}\bfM^n\right)\bfL \\
 &= \bfL(\bfI_n - \bfM)^{-1}\bfL.
\end{split}
\end{equation}
\end{proof}

\begin{lemma}\label{lem:boderline_as}
Let the random variable $\bfa$ have a have a density with respect to the Lebesgue measure on $(\R_0^+)^n$. 
The set of all $\mathbf{a}$ for which the system solution contains at least one borderline firm, 
i.e. one $i\in\mathcal N$ such that $r_i^* = d_i$ and $s_i^* = 0$, then has measure zero.
\end{lemma}

\begin{proof}
First, note that it suffices to show the claim for a set $A(I)$ of all $\mathbf{a}$ for which $r_i^*=d_i$ and $s_i^*=0$ 
for each $i \in I \subset \mathcal N$, since the number of subsets of $\{1,...,n\}$ is finite 
and a finite union of sets of Lebesgue measure zero has Lebesgue measure zero.
We first show that $A(I)$ is a Borel set and hence Lebesgue measurable. 
For this, note that it is shown in \citet{fischer14} that the mapping $\Psi : \mathbf{a} \mapsto \mathbf{R}^*(\mathbf{a})$
that maps any price vector of the exogenous assets onto the corresponding solution of 
\eqref{eq:liq_eq_debt} and \eqref{eq:liq_eq_equity} is Borel measurable. 
Let now $H(I)$ denote the $2(n-|I|)$-dimensional hyperplane in $\mathbb{R}^{2n}$ for which
\begin{equation}
H(I) = \{(\bfr^t, \bfs^t)^t\in \R^{2n}: r_i = d \text{ and } s_i=0 \text{ for all } i\in I\} .
\end{equation}
Clearly, $H(I)$ is a Borel set. One obtains
\begin{equation}
A(I) = \Psi^{-1}(H(I) \cap (\mathbb{R}^+_0)^{2n}) ,
\end{equation}
which must be Borel-measurable, too. 
Observe now that if $\mathbf{a}_2 \gg \mathbf{a}_1$ ($\mathbf{a}_2$ strictly larger than $\mathbf{a}_1$ in all components), 
then $\Phi^n_{\mathbf{a}_2}(\mathbf{R}) \ge \Phi^n_{\mathbf{a}_1}(\mathbf{R})$ for any non-negative $\mathbf{R}$. 
Hence, by the Picard Iteration, $\mathbf{R}^*(\mathbf{a }_2) \geq \mathbf{R}^*(\mathbf{a}_1)$. 
From Eq. \eqref{eq:liq_eq_debt} and \eqref{eq:liq_eq_equity}, it follows that $\bfr + \bfs = \bfa + \bfMd\bfr + \bfMs\bfs$ 
(see also \citet{fischer14}). 
Therefore, if $\mathbf{a}_2 \gg \mathbf{a}_1$ , then $r_i^*(\bfa_2) + s_i^*(\bfa_2) > r_i^*(\bfa_1) + s_i^*(\bfa_1)$ for all $i\in\mathcal N$,
which is a contradiction to $\mathbf{a}_1, \mathbf{a}_2 \in A(I)$. 
Thus, since $r_i^*(\bfa) + s_i^*(\bfa) = d_i$ for $\bfa\in A(I)$ and $i\in I$, 
$\mathbf{a}_2 \gg \mathbf{a}_1$ can hold for no pair $\mathbf{a}_1, \mathbf{a}_2 \in A(I)$. 
This means that $A(I)$ bears some resemblance to a Pareto set. 
It follows that the set $A(I)$ intersects any straight line parallel to the vector $(1,\ldots,1)^t\in\R^n$ either once, or not at all. 
As such, and since the Lebesgue measure is rotation invariant, the problem reduces now to the one which is shown 
in Lemma \ref{lem:borel_set_measure_zero}.
\end{proof}

\begin{lemma}\label{lem:borel_set_measure_zero}
Let $Q$ be a Borel set in $\mathbb{R}^n$ such that $|Q_{\omega}|\leq 1$ for any $\omega\in\mathbb{R}^{n-1}$, 
where $Q_{\omega}=\{x\in\mathbb{R}: (x,\omega^t)^t\in Q\}$. 
Then $Q$ has Lebesgue measure zero.
\end{lemma}

\begin{proof}
Let $\lambda_{m}, m\in\mathbb N$, denote the Lebesgue measure on $\mathbb{R}^{m}$.
For any Borel set $Q$, it follows from the definition of product measures 
(e.g.~\citet{billingsley95}) and from $\lambda_{n}=\lambda_{1}\otimes\lambda_{n-1}$ that
\begin{equation}
\lambda_{n}(Q) = \int \lambda_{1}(Q_\omega) d \lambda_{n-1}(\omega) .
\end{equation}
Since $\lambda_{1}(Q_\omega)=0$, the result follows.
\end{proof}

\begin{table}[htbp]
\centering
\caption{Above: Mean runtime in seconds for every algorithm over all debt and equity integration values ($\nu^{\bfd}$ and $\nu^{\bfd}$)
and system sizes ($n$) grouped by the debt values ($d$). 
Below: Mean runtime in seconds for every algorithm over all equity integration values ($\nu^{\bfs}$), debt values ($d$)
and system sizes ($n$) grouped by the debt integration values ($\nu^{\bfd}$).
}
\label{tab:res_sim_debt_int}
\begin{tabular}{crrrrrrrrrr}\toprule
                              &      & \multicolumn{9}{c}{Algorithm}                                        \\\cline{3-11}
                              &      &   P  &    E  &    H  &   TP  &   TE  &   TH  &   SP  &   SE  &   SH  \\\midrule
\multirow{5}{*}{$d$}          & 1    & 5.43 & 31.21 & 33.87 &  9.51 & 22.71 & 26.78 & 10.94 & 30.90 & 35.17 \\
                              & 1.5  & 5.61 & 35.73 & 36.59 &  9.95 & 25.03 & 28.61 & 11.86 & 38.27 & 39.72 \\
                              & 2    & 5.98 & 36.46 & 34.70 & 10.31 & 25.00 & 27.56 & 12.55 & 40.92 & 39.19 \\
                              & 2.5  & 6.46 & 36.67 & 31.69 & 10.53 & 24.65 & 26.23 & 12.97 & 41.04 & 36.29 \\
                              & 3    & 6.81 & 36.87 & 28.49 & 10.56 & 23.53 & 24.07 & 12.80 & 40.22 & 32.67 \\\midrule
\multirow{7}{*}{$\nu^{\bfd}$} & 0.05 & 4.04 & 17.24 & 20.85 &  8.86 & 13.21 & 18.79 &  9.13 & 15.25 & 20.49 \\
                              & 0.2  & 5.34 & 26.56 & 27.01 &  9.41 & 17.04 & 23.52 & 10.42 & 22.60 & 28.80 \\
                              & 0.35 & 6.38 & 36.29 & 32.89 & 10.14 & 22.30 & 27.12 & 11.95 & 32.59 & 36.41 \\
                              & 0.5  & 7.40 & 47.30 & 40.33 & 11.44 & 30.18 & 31.53 & 13.98 & 47.05 & 44.93 \\
                              & 0.65 & 7.45 & 55.41 & 46.72 & 11.75 & 35.70 & 34.44 & 15.24 & 61.52 & 51.42 \\
                              & 0.8  & 6.21 & 39.27 & 37.51 & 10.23 & 28.01 & 28.20 & 13.31 & 50.73 & 41.69 \\
                              & 0.95 & 5.58 & 25.65 & 26.16 &  9.37 & 22.84 & 22.95 & 11.53 & 38.17 & 32.53 \\\bottomrule
\end{tabular}
\end{table}

\subsection{Additional Tables and Simulation Results}

The notation in the tables in this section is as follows. 
The names of the algorithms in the table are composed out of their iteration type (``P'' for Picard, ``E'' for Elsinger and ``H'' for Hybrid) and 
their direction (``D'' for decreasing, ``I'' for increasing). 
If the prefix ``D'' or ``I'' is omitted, the mean value of the corresponding increasing and decreasing version is shown.
The additional prefix ``T'' denotes the Trial-and-Error version and ``S'' denotes the Sandwich version of the algorithm.

\pagebreak

\begin{table}[htbp]
\centering
\caption{Median of the iteration (calculation) steps for every algorithm over all debt and equity integration values 
($\nu^{\bfd}$ and $\nu^{\bfd}$) and debt values ($d$) grouped by the system size ($n$). 
A calculation step is defined as the solution of a linear equation system, 
which is done for example in Algorithm \ref{alg:equity} in every iteration step. 
An iteration step is defined as the step from the $k$-th iterate $\bfR^k$ to $\bfR^{k+1}$ for $k\ge0$, no matter which algorithm is used.}
\label{tab:res_sim_iter_calc}
\begin{tabular}{lllllll}\toprule
    & \multicolumn{6}{c}{system size $n$}                     \\\cline{2-7}
    & 5     & 10     & 25     & 50      & 100      & 200      \\\midrule
DP  & 6     &  7     &  8     &  8      &   9      &   9      \\
IP  & 8     &  9     & 10     & 10      &  11      &  11      \\
DE  & 4 (5) &  5 (6) &  5 (8) &  5 (9)  &   5 (10) &   6 (12) \\
IE  & 4 (6) &  5 (7) &  6 (9) &  6 (11) &   7 (12) &   7 (13) \\
DH  & 2 (5) &  3 (7) &  3 (8) &  3 (9)  &   3 (11) &   3 (11) \\
IH  & 2 (4) &  3 (4) &  3 (4) &  3 (5)  &   3 (6)  &   3 (8)  \\\midrule
DTP & 2 (1) &  2 (1) &  2 (1) &  3 (1)  &   3 (1)  &   3 (1)  \\
ITP & 2 (1) &  3 (1) &  3 (1) &  4 (1)  &   4 (1)  &   5 (1)  \\
DTE & 1 (3) &  1 (3) &  1 (5) &  1 (5)  &   2 (5)  &   2 (7)  \\
ITE & 1 (3) &  1 (4) &  2 (5) &  2 (7)  &   2 (7)  &   3 (8)  \\
DTH & 1 (4) &  1 (4) &  1 (6) &  1 (6)  &   2 (7)  &   2 (8)  \\
ITH & 1 (3) &  1 (4) &  2 (4) &  2 (5)  &   2 (6)  &   2 (7)  \\\midrule
SP  & 1     &  1     &  2     &  2      &   3      &   3      \\
SE  & 0 (4) &  1 (4) &  1 (6) &  1 (8)  &   2 (8)  &   2 (12) \\
SH  & 0 (4) &  1 (4) &  1 (5) &  1 (6)  &   1 (8)  &   1 (7)  \\\bottomrule
\end{tabular}
\end{table}

\bibliographystyle{plainnat}
\bibliography{literature}

\end{document}